\documentclass[runningheads]{llncs}

\usepackage{graphicx}
\usepackage{hyperref}
\usepackage{xcolor}
\usepackage{textcomp}
\usepackage{bbding}
\usepackage{amsmath}
\usepackage{amssymb}
\usepackage{pgf}
\usepackage{tikz}
\usetikzlibrary{arrows, automata, shapes, petri, positioning, calc}
\usepackage{rotating}
\usepackage{float}
\usepackage[ruled,vlined,linesnumbered]{algorithm2e}
\usepackage{subcaption}
\usepackage{booktabs}
\usepackage[labelfont=bf,tableposition=top]{caption}
\usepackage{multirow}
\SetKw{Continue}{continue}
\SetKw{Break}{break}

\hypersetup{
	colorlinks,
	linkcolor={red!50!black},
	citecolor={blue!60!black},
	urlcolor={blue!80!black}
}

\def\bag{{\cal B}}
\def\nomove{\gg}

\begin{document}

\title{Efficient Time and Space Representation\\of Uncertain Event Data\thanks{
		We thank the Alexander von Humboldt (AvH) Stiftung for supporting our research interactions.
		Please do not print this document unless strictly necessary.
}}

\author{Marco Pegoraro~\Envelope\orcidID{0000-0002-8997-7517} \and Merih Seran Uysal\orcidID{0000-0003-1115-6601} \and Wil M.P. van der Aalst\orcidID{0000-0002-0955-6940}}

\authorrunning{Pegoraro et al.}

\institute{Process and Data Science Group (PADS) \\ Department of Computer Science, RWTH Aachen University, Aachen, Germany
	\email{\{pegoraro, uysal, wvdaalst\}@pads.rwth-aachen.de}\\
	\url{http://www.pads.rwth-aachen.de/}}

\maketitle

\begin{abstract}
Process mining is a discipline which concerns the analysis of execution data of operational processes, the extraction of models from event data, the measurement of the conformance between event data and normative models, and the enhancement of all aspects of processes. Most approaches assume that event data is accurately capture behavior. However, this is not realistic in many applications: data can contain uncertainty, generated from errors in recording, imprecise measurements, and other factors. Recently, new methods have been developed to analyze event data containing uncertainty; these techniques prominently rely on representing uncertain event data by means of graph-based models explicitly capturing uncertainty. In this paper, we introduce a new approach to efficiently calculate a graph representation of the behavior contained in an uncertain process trace. We present our novel algorithm, prove its asymptotic time complexity, and show experimental results that highlight order-of-magnitude performance improvements for the behavior graph construction.

\end{abstract}

\section{Introduction}\label{sec:introduction}
The pervasive diffusion of digitization, which gained momentum thanks to advancements in electronics and computing at the end of the last century, brought a wave of innovation in the tools supporting businesses and companies. The past decades have seen the rise of \emph{Process-Aware Information Systems} (PAISs) -- useful to structurally support processes in a business -- as well as research disciplines such as Business Process Management (BPM) and \emph{process mining}.

Process mining~\cite{van2016process} is a field of research that enables process analysis in a data-driven manner. Process mining analyses are based on recordings of tasks and events in a process, memorize in an ensemble of information systems which support business operations. These recordings are exported and systematically collected in databases called \emph{event logs}. Using an event log as a starting point, process mining techniques can automatically obtain a process model illustrating the behavior of the real-life process (\emph{process discovery}) and identify anomalies and deviations between the execution data of a process and a normative model (\emph{conformance checking}). Process mining is a subfield of data science which is quickly growing in interest both in academia and industry. Over 30 commercial software tools are available on the market for analyzing processes and their execution data. Process mining tools are used by process experts to analyze processes in tens of thousands of organizations, e.g., within Siemens, over 6000 employees actively use process mining to improve internal procedures.

Commercial process mining software is able to discover and build a process model from an event log. Most of the process discovery algorithms implemented in these tools are based on tallying the number of \emph{directly-follows relationships} between activities in the execution data of the process. The more frequently a specific activity immediately follows another one in the execution log of a process, the stronger a causality and/or precedence implication between the two activities is understood to be. Such directly-follows relationships are also the basis for the identification of more complex and abstract constructs in the workflow of a process, such as interleaving or parallelism of activities. These relationships between activities are often represented in a labeled directed graph called the \emph{Directly-Follows Graph} (DFG).

In recent times, a new type of event logs has gained research interest: \emph{uncertain event logs}~\cite{pegoraro2019mining}. Such execution logs contain, rather than precise values, an indication of the possible values that event attributes can acquire. In this paper, we will consider the setting where uncertainty is represented by either an interval or a set of possible values for an event attribute. Moreover, we will consider the case in which an event has been recorded in the event log albeit it did not happen in reality.

Uncertainty in event logs is best illustrated with a real-life example of a process that can generate uncertain data in an information system. Let us consider the following process instance, a simplified version of anomalies that are actually occurring in processes of the healthcare domain. An elderly patient enrolls in a clinical trial for an experimental treatment against myeloproliferative neoplasms, a class of blood cancers. The enrollment in this trial includes a lab exam and a visit with a specialist; then, the treatment can begin. The lab exam, performed on the 8th of July, finds a low level of platelets in the blood of the patient, a condition known as thrombocytopenia (TP). At the visit, on the 10th of May, the patient self-reports an episode of night sweats on the night of the 5th of July, prior the lab exam: the medic notes this, but also hypothesized that it might not be a symptom, since it can be caused not by the condition but by external factors (such as very warm weather). The medic also reads the medical records of the patient and sees that, shortly prior to the lab exam, the patient was undergoing a heparine treatment (a blood-thinning medication) to prevent blood clots. The thrombocytopenia found with the lab exam can then be primary (caused by the blood cancer) or secondary (caused by other factors, such as a drug). Finally, the medic finds an enlargement of the spleen in the patient (splenomegaly). It is unclear when this condition has developed: it might have appeared in any moment prior to that point. The medic decides to admit the patient in the clinical trial, starting 12th of July.

These events generate the trace of Table~\ref{table:uncertaintrace} in the information system of the hospital. For clarity, the timestamp field only reports the day of the month.

\begin{table}[]
	\caption{The uncertain trace of an instance of healthcare process used as running example. The ``Case ID'' is a unique identifier for all events in a single process case; the ``Event ID'' is a unique identifier for the events in the trace. The ``Timestamp'' field indicates either the moment in time in which the event has happened, or the interval of time in which the event may have happened. The ``Activity'' field indicates the possible choices for the activity instantiated by the event. Lastly, the ``Indeterminate event'' field contains a ``!'' if the corresponding event has surely occurred, and a ``?'' if it might have been recorded despite not occurring in reality. For the sake of readability, in the timestamps column only reports the day of the month.}
	\label{table:uncertaintrace}
	\centering
	\begin{tabular}{ccccc}
		\textbf{Case ID}        & \textbf{Event ID} & \textbf{Timestamp}                                                                                                     & \textbf{Activity}             & \multicolumn{1}{l}{\textbf{Indet. event}} \\ \hline
		\multicolumn{1}{|c|}{ID327} & \multicolumn{1}{c|}{$e_1$} 
		& \multicolumn{1}{c|}{5}                                                                         & \multicolumn{1}{c|}{\emph{NightSweats}}        & \multicolumn{1}{c|}{?}                    \\ \hline
		\multicolumn{1}{|c|}{ID327}& \multicolumn{1}{c|}{$e_2$} & \multicolumn{1}{c|}{8}                                                                         & \multicolumn{1}{c|}{\{\emph{PrTP}, \emph{SecTP}\}} & \multicolumn{1}{c|}{!}                    \\ \hline
		\multicolumn{1}{|c|}{ID327}& \multicolumn{1}{c|}{$e_3$} & \multicolumn{1}{c|}{[4, 10]}                                                                         & \multicolumn{1}{c|}{\emph{Splenomeg}} & \multicolumn{1}{c|}{!}                    \\ \hline
		\multicolumn{1}{|c|}{ID327}& \multicolumn{1}{c|}{$e_4$} & \multicolumn{1}{c|}{12}                                                                         & \multicolumn{1}{c|}{\emph{Adm}}        & \multicolumn{1}{c|}{!}                    \\ \hline
	\end{tabular}
\end{table}

Event $e_2$ has been recorded with two possible activity labels (\emph{PrTP} or \emph{SecTP}). This is an example of uncertainty on activities. Some events, e.g. $e_3$, do not have a precise timestamp but a time interval in which the event could have happened has been recorded: in some cases, this causes the loss of a precise ordering of events (e.g. $e_2$ and $e_3$). This is an instance of uncertainty on the time dimension, i.e., on timestamps. As evident by the ``?'' symbol, $e_1$ is an indeterminate event: it has been recorded, but it is not guaranteed to have actually happened. Conversely, the ``!'' symbol indicates that the event has been recorded while certainly occurring in reality, i.e., it has been recorded correctly in the information system (e.g., the event $e_4$).

Quality problems and imprecision in data recording such as the ones described in the running example as source of uncertainty are not uncommon; in some settings, they are a frequent occurrence. Healthcare processes are specifically know to be afflicted by these sorts of data anomalies, especially if parts of the process rely on recording information on paper~\cite{van2011process, kurniati2019assessment}. Existing process mining software cannot manage such uncertain event data. When mining the processes where uncertainty in execution data is prominent, a natural first approach is to filter the event log eliminating cases where uncertainty appear. Unfortunately, in processes with a large portion of cases are affected by such data anomalies, filtering without losing essential information about the process is not feasible.

As a consequence, new process mining methods to inspect and analyze it have to be developed. Uncertain timestamps are the most prominent and critical source of uncertain behavior in a process trace. For example, if $n$ events have uncertain timestamps such that their order is unknown, the possible configurations that the control-flow of the trace can assume are all the $n!$ permutations of the events, in the case where all events in a case have timestamps defined by mutually overlapping intervals. This is the worst possible scenario in terms of amount of uncertain behavior introduced by uncertainty on the timestamps of the events ins a trace. Thus, it is important to capture the time relationships between events in a compact and effective way. This is accomplished by the construction of a \emph{behavior graph}, a directed acyclic graph that expresses precedence between events. Figure~\ref{fig:bg} shows the behavior graph of the process trace in Table~\ref{table:uncertaintrace}; every known precedence relationship between events is represented by the edges of the graph, while the pairs of event for which the order is unknown remain unconnected. Effectively, this creates a representation of the partial order where the arcs are defined by the possible values of the timestamps contained in the trace, and where the nodes may refer to sets of possible activities. As we will see, this construct is central to effectively implement both process discovery and conformance checking applied to uncertain event data.

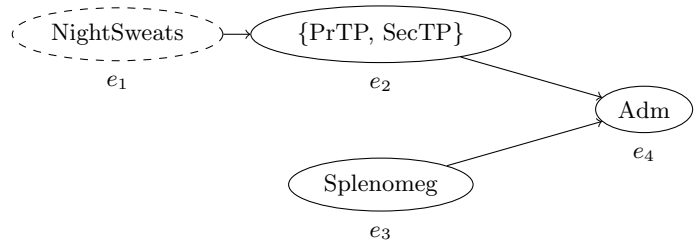
\begin{figure}
	\centering
	\begin{tikzpicture}[->, node distance=3.5cm, nodes={draw, ellipse}]
	
	\node[dashed]	(A)	[label=below:$e_1$]								{$\text{NightSweats}$};
	\node			(B)	[right of=A, label=below:$e_2$]					{$\{\text{PrTP, SecTP}\}$};
	\node			(C)	[below of=B, yshift=1.5cm, label=below:$e_3$]	{$\text{Splenomeg}$};
	\node			(D)	[right of=C, yshift=1cm, label=below:$e_4$]		{$\text{Adm}$};
	
	\path
	(A) edge (B)
	(B) edge (D)
	(C) edge (D);
	\end{tikzpicture}
	\caption{The behavior graph of the trace in Table~\ref{table:uncertaintrace}. Every node represents an event; the labels in the nodes represent the activity, or set of activities, associated with the event. The arcs represent the partial order relationship between events as defined by their timestamps. The indeterminate event, which might not have occurred, is represented by a dashed node.}
	\label{fig:bg}
\end{figure}

In a previous paper~\cite{pegoraro2020efficient}, we presented a time-effective algorithm for the construction of the behavior graph of an uncertain process trace, attaining quadratic time complexity on the number of events in the trace.

This paper elaborates on this previous result, by providing the proof of the correctness of the new algorithm. Additionally, we will show the improvement in performance both theoretically, via asymptotic complexity analysis, and in practice, with experiments on various uncertain event logs comparing computation times of the baseline method against the novel construction algorithm. Furthermore, the version of the algorithms presented in this paper is refined so to preprocess uncertain traces in linear time, individuating the variants -- which share the same behavior graph --, and proceed to perform the construction of the behavior graph only once per variant. This slightly improves performance, and more importantly, enables the representation of an uncertain event log as a multiset of behavior graphs, greatly reducing the memory requirements to store the log. This enables a streamlined application of process mining techniques on event data where uncertainty is present.

The algorithms have been implemented within the PROVED (\emph{PRocess mining OVer uncErtain Data}) library\footnote{\url{https://github.com/proved-py/proved-core/tree/Efficient_Time_and_Memory_Representation_for_Uncertain_Event_Data}}, based on the PM4Py process mining framework~\cite{berti2019process}.

The reminder of the paper is structured as follows: Section~\ref{sec:conf} motivates the study of uncertainty in process mining by illustrating an example of conformance checking over uncertain event data. Section~\ref{sec:disc} strengthens the motivation showing the discovery of process models of uncertain event logs. Section~\ref{sec:materialsandmethods} provides formal definitions, describes the baseline technique for our research, and shows a new and more efficient method to obtain a behavior graph of an uncertain trace. Section~\ref{sec:proofs} presents the analysis of asymptotic complexity for both the baseline and the novel method. Section~\ref{sec:experiments} shows results of experiments on both synthetic and real-life uncertain event logs comparing the efficiency of both methods to compute behavior graphs. Section~\ref{sec:related} explores recent related works in the context of uncertain event data and the management of alterations of data in process mining. Finally, Section~\ref{sec:conclusions} discusses the output of the experiments and concludes the paper.

\section{Conformance Checking over Uncertain Data}\label{sec:conf}

Conformance checking is one of the main tasks in process mining, and consists in measuring the deviation between process execution data (usually in the form of a trace) and a reference model. This is particularly useful for organization, since it enables them to compare historical process data against a normative model created by process experts to identify anomalies and deviations in their operations.

Let us assume that we have access to a normative model for the disease of the patient in the running example, shown in Figure~\ref{fig:samplemodel}.

\begin{figure}
	\centering
	\begin{tikzpicture}[node distance=.7cm and .3cm, >=stealth']
	
	\tikzstyle{place} = [circle,draw,thick,minimum size=6mm]
	\tikzstyle{transition} = [rectangle,draw,thick,minimum size=4mm]
	\tikzstyle{invisible} = [transition, fill=black]
	\tikzstyle{finaltoken} = [token, fill=black!30]
	
	\node [place,tokens=1] (p1) {};
	
	\node [invisible] (t0) [right= of p1, label=above:{\scriptsize $t_1$}] {};
	\draw [->] (p1) to (t0.west);
	
	\node [place] (p11) [above right= of t0] {};
	\draw [->] (t0.east) to (p11);
	
	\node [place] (p12) [below right= of t0] {};
	\draw [->] (t0.east) to (p12);
	
	\node [transition] (t1) [right= of p11, label=above:{\scriptsize $t_2$}] {NightSweats};
	\draw [->] (p11) to (t1.west);
	
	\node [place] (p21) [right= of t1] {};
	\draw [->] (t1.east) to (p21);
	
	\node [invisible] (t01) [below right= of p21, label=above:{\scriptsize $t_4$}] {};
	\draw [->] (p21) to (t01.west);
	
	\node [place] (p22) [below left= of t01] {};
	
	\draw [->] (p22) to (t01.west);
	
	\node [transition] (t2) at ($(p12)!0.5!(p22)$) [label=above:{\scriptsize $t_3$}] {Splenomeg};
	\draw [->] (p12) to (t2.west);
	
	\draw [->] (t2.east) to (p22);
	
	\node [place] (p2) [right= of t01] {};
	\draw [->] (t01.east) to (p2);
	
	\node [transition] (t4) [right= of p2, label=above:{\scriptsize $t_5$}] {PrTP};
	\draw [->] (p2) to (t4.west);
	
	\node [place] (p3) [right= of t4] {};
	\draw [->] (t4.east) to (p3);
	
	\node [transition] (t6) [right= of p3, label=above:{\scriptsize $t_6$}] {Adm};
	\draw [->] (p3) to (t6);
	
	\node [place] (p6) [right= of t6] {};
	\draw [->] (t6) to (p6);
	\node [finaltoken] at (p6) {};
	\end{tikzpicture}
	\caption{A normative model for the healthcare process case in the running example. The initial marking is displayed; the gray ``token slot'' represents the final marking.}
	\label{fig:samplemodel}
\end{figure}
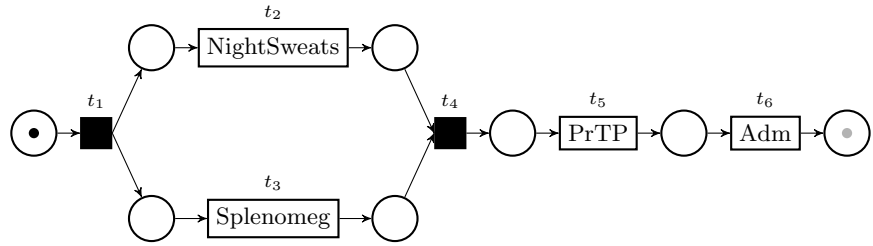

This model essentially states that the disease is characterised by the occurrence of night sweats and splenomegaly on the patient, which can be verified concurrently, and then should be followed by primary thrombocytopenia. We would like to measure the conformance between the trace in Table~\ref{table:uncertaintrace} and this normative model. A very popular conformance checking technique works via the computation of \emph{alignments}~\cite{adriansyah2010towards}. Through this technique, we are able to identify the deviations in the execution of a process, in the form of behavior happening in the model but not in the trace, and behavior happening in the trace but not in the model. These deviations are identified, and used as basis to compute a conformance score between the trace and the process model.

The formulation of alignments in~\cite{adriansyah2010towards} is not applicable to an uncertain trace. In fact, depending on the instantiation of the uncertain attributes of events -- like the timestamp of $e_3$ in the trace -- the order of event may differ, and so may the conformance score. However, we can look at the best- and worst case-scenarios: the instantiation of attributes of the trace that entails the minimum and maximum number of deviations with respect to the reference model. In our example, two possible outcomes for the sample trace are $\langle \textit{NightSweats}, \textit{Splenomeg}, \textit{PrTP}, \textit{Adm} \rangle$ and $\langle \textit{SecTP}, \textit{Splenomeg}, \textit{Adm} \rangle$; both represent the sequence of event that might have happened in reality, but their conformance score is very different. The alignment of the first trace against the reference model can be seen in Table~\ref{table:bestalign}, while the alignment of the second trace can be seen in Table~\ref{table:worstalign}. These two outcomes of the uncertain trace in Table~\ref{table:uncertaintrace} represent, respectively, the minimum and maximum amount of deviation possible with respect to the reference model, and define then a lower and upper bound for conformance score.

\begin{table}[]
	\caption{An optimal alignment for $\langle \textit{NightSweats}, \textit{Splenomeg}, \textit{PrTP}, \textit{Adm} \rangle$, one of the possible instantiations of the trace in Table~\ref{table:uncertaintrace}, against the model in Figure~\ref{fig:samplemodel}. This alignment has a deviation cost equal to 0, and corresponds to the best case scenario for conformance between the process model and the uncertain trace.}
	\label{table:bestalign}
	\centering
	\begin{tabular}{cccccc}
		\multicolumn{1}{|c|}{$\nomove$}	& \multicolumn{1}{c|}{NightSweats}	& \multicolumn{1}{c|}{Splenomeg}	& \multicolumn{1}{c|}{$\nomove$}	& \multicolumn{1}{c|}{PrTP}		& \multicolumn{1}{c|}{Adm}		\\ \hline
		\multicolumn{1}{|c|}{$\tau$}	& \multicolumn{1}{c|}{NightSweats}	& \multicolumn{1}{c|}{Splenomeg}	& \multicolumn{1}{c|}{$\tau$}	& \multicolumn{1}{c|}{PrTP}		& \multicolumn{1}{c|}{Adm}		\\ 
		\multicolumn{1}{|c|}{$t_1$}		& \multicolumn{1}{c|}{$t_2$}	& \multicolumn{1}{c|}{$t_3$}	& \multicolumn{1}{c|}{$t_4$}				& \multicolumn{1}{c|}{$t_5$}	& \multicolumn{1}{c|}{$t_6$}	\\ 
	\end{tabular}
\end{table}

\begin{table}[]
	\caption{An optimal alignment for $\langle \textit{SecTP}, \textit{Splenomeg}, \textit{Adm} \rangle$, one of the possible instantiations of the trace in Table~\ref{table:uncertaintrace}, against the model in Figure~\ref{fig:samplemodel}. This alignment has a deviation cost equal to 3, caused by 2 moves on model and 1 move on log, and corresponds to the worst case scenario for conformance between the process model and the uncertain trace.}
	\label{table:worstalign}
	\centering
	\begin{tabular}{ccccccc}
		\multicolumn{1}{|c|}{$\nomove$}& \multicolumn{1}{c|}{SecTP}	& \multicolumn{1}{c|}{$\nomove$}	& \multicolumn{1}{c|}{Splenomeg}	& \multicolumn{1}{c|}{$\nomove$}	& \multicolumn{1}{c|}{$\nomove$}				& \multicolumn{1}{c|}{Adm}		\\ \hline
		\multicolumn{1}{|c|}{$\tau$}	& \multicolumn{1}{c|}{$\nomove$}	& \multicolumn{1}{c|}{NightSweats}		& \multicolumn{1}{c|}{Splenomeg}	& \multicolumn{1}{c|}{$\tau$}	& \multicolumn{1}{c|}{PrTP}			& \multicolumn{1}{c|}{Adm}	\\ 
		\multicolumn{1}{|c|}{$t_1$}		& \multicolumn{1}{c|}{}			& \multicolumn{1}{c|}{$t_2$}		& \multicolumn{1}{c|}{$t_3$}	& \multicolumn{1}{c|}{$t_4$}				& \multicolumn{1}{c|}{$t_5$}		& \multicolumn{1}{c|}{$t_6$}		\\ 
	\end{tabular}
\end{table}

The minimum and maximum bounds for conformance score of an uncertain trace and a reference process model can be found with the uncertain version of the alignment technique that we first described in~\cite{pegoraro2019mining}. In order to find such bounds, it is necessary to build a Petri net able to simulate all possible behaviors in the uncertain trace, called the \emph{behavior net}. Obtaining a behavior net is possible through a construction that uses behavior graphs as a starting point, using the structural information therein contained to connect places and transitions in the net. The behavior net of the trace in Table~\ref{table:uncertaintrace} is shown in Figure~\ref{fig:bn}.

\begin{figure}
	\centering
	\begin{tikzpicture}[node distance=.7cm and .9cm, >=stealth']
	
	\tikzstyle{place} = [circle,draw,thick,minimum size=6mm]
	\tikzstyle{transition} = [rectangle,draw,thick,minimum size=4mm]
	\tikzstyle{invisible} = [transition, fill=black]
	\tikzstyle{finaltoken} = [token, fill=black!30]
	
	\node [place,tokens=1] (p1) [label=above:{\scriptsize $(\textsc{start}, e_1)$}] {};
	
	\node [transition] (t1) [above right= of p1, label=above:{\scriptsize $(e_1, NightSweats)$}] {NightSweats};
	\draw [->] (p1) to (t1.west);
	
	\node [invisible] (t2) [below right= of p1, label=above:{\scriptsize $(e_1, \tau)$}] {NightSweats};
	\draw [->] (p1) to (t2.west);
	
	\node [place] (p2) [below right= of t1, label=above:{\scriptsize $(e_1, e_2)$}] {};
	\draw [->] (t1.east) to (p2);
	\draw [->] (t2.east) to (p2);
	
	\node [transition] (t3) [above right= of p2, label=above:{\scriptsize $(e_2, PrTP)$}] {PrTP};
	\draw [->] (p2) to (t3.west);
	
	\node [transition] (t4) [below right= of p2, label=above:{\scriptsize $(e_2, SecTP)$}] {SecTP};
	\draw [->] (p2) to (t4.west);
	
	\node [place] (p3) [below right= of t3, label=above:{\scriptsize $(e_2, e_4)$}] {};
	\draw [->] (t3.east) to (p3);
	\draw [->] (t4.east) to (p3);
	
	\node [place,tokens=1] (p4) [below left= of t2, label=above:{\scriptsize $(\textsc{start}, e_3)$}] {};
	
	\node [place] (p5) [below right= of t4, label=above:{\scriptsize $(e_3, e_4)$}] {};
	
	\node [transition] (t5) at ($(p4)!0.5!(p5)$) [label=above:{\scriptsize $(e_3, Splenomeg)$}] {Splenomeg};
	\draw [->] (p4) to (t5);
	\draw [->] (t5) to (p5);
	
	\node [transition] (t6) [above right= of p5, label=above:{\scriptsize $(e_4, Adm)$}] {Adm};
	\draw [->] (p3) to (t6.north west);
	\draw [->] (p5) to (t6.south west);
	
	\node [place] (p6) [right= of t6, label=above:{\scriptsize $(e_4, \textsc{end})$}] {};
	\draw [->] (t6) to (p6);
	\node [finaltoken] at (p6) {};
	\end{tikzpicture}
	\caption{The behavior net representing the behavior of the uncertain trace in Table~\ref{table:uncertaintrace} and obtained thanks to its behavior graph. The initial marking is displayed; the gray ``token slot'' represents the final marking. This artifact is necessary to perform conformance checking between uncertain traces and a reference model.}
	\label{fig:bn}
\end{figure}
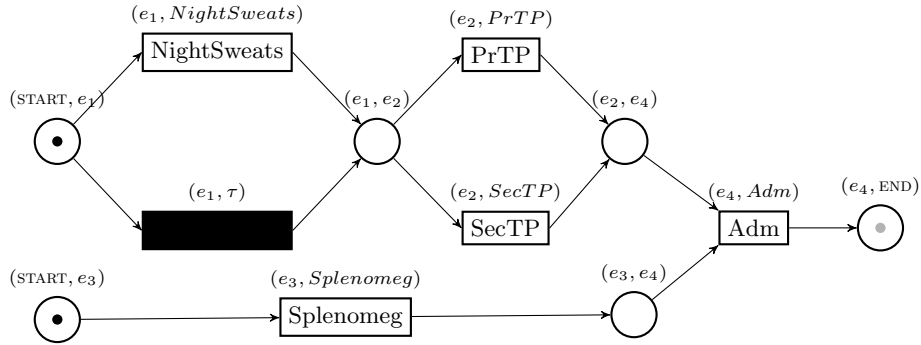

The alignments in Tables~\ref{table:bestalign} and~\ref{table:worstalign} show how we can get actionable insights from process mining over uncertain data. In some applications it is reasonable and appropriate to remove uncertain data from an event log via filtering, and then compute log-level aggregate information -- such as total number of deviations, or average deviations per trace -- using the remaining certain data. Even in processes where this is possible, doing so prevents the important process mining task of case diagnostic. Conversely, uncertain alignments allow not only to have best- and worst-case scenarios for a trace, but also to individuate the specific deviations affecting both scenarios. For instance, the alignments of the running example can be implemented in a system that warns the medics that the patient might have been affected by a secondary thrombocytopenia not explained by the model of the disease. Since the model indicates that the disease should develop primary thrombocytopenia as a symptom, this patient is at risk of both types of platelets deficit simultaneously, which is a serious situation. The medics can then intervene to avoid this complication, and performing more exams to ascertain the cause of the patient's thrombocytopenia.

\section{Process Discovery over Uncertain Data}\label{sec:disc}

Process discovery is another main objective in process mining, and involves automatically creating a process model from event data. Many process discovery algorithms rely on the concept of \emph{directly-follows relationships} between activities to gather clues on how to structure the process model. \emph{Uncertain Directly-Follows Graphs} (UDFGs) enable the representation of directly-follows relationships in an event log under conditions of uncertainty in the event data; they consist in directed graphs where the activity labels appearing in an event log constitute the nodes, and the edges are decorated with information on the minimum and maximum frequency observable for the directly-follows relation between pair of activities.

Let us examine an example of UDFG. In order to build a significant example, we need to introduce an entire uncertain event log; since the full table notation for uncertain traces becomes cumbersome for entire logs, let us utilize a shorthand simplified notation. In a trace, we represent an uncertain event with multiple possible activity labels by listing all the associated labels between curly braces.

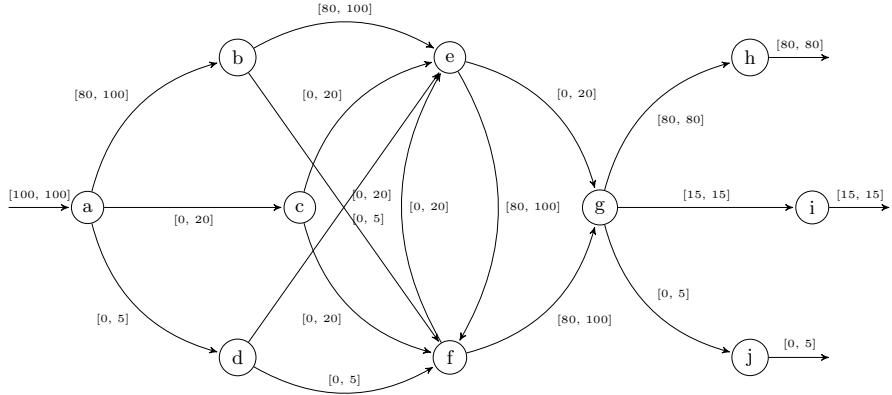
\begin{figure}[h]
	\centering
	\resizebox{\textwidth}{!}{%
	\begin{tikzpicture}[->,>=stealth',shorten >=1pt,node distance=3.4cm,auto,main node/.style={circle,draw,align=center}]
	\node[main node]	(A)	[]						{a};
	\node[main node]	(B)	[above right of=A]		{b};
	\node[main node]	(C)	[right of=A]			{c};
	\node[main node]	(D)	[below right of=A]		{d};
	\node[main node]	(E)	[above right of=C]		{e};
	\node[main node]	(F)	[below right of=C]		{f};
	\node[main node]	(G)	[below right of=E]		{g};
	\node[main node]	(H)	[above right of=G]		{h};
	\node[main node]	(I)	[right of=G]			{i};
	\node[main node]	(J)	[below right of=G]		{j};
	\node[left=1cm of A] (K) {};
	\node[right=1cm of H] (L) {};
	\node[right=1cm of I] (M) {};
	\node[right=1cm of J] (N) {};
	\path
	(A) edge [bend left] node {\tiny [80, 100]} (B)
	(A) edge node [swap] {\tiny [0, 20]} (C)
	(A) edge [bend right] node [swap] {\tiny [0, 5]} (D)
	(B) edge [bend left] node {\tiny [80, 100]} (E)
	(B) edge node {\tiny [0, 20]} (F)
	(C) edge [bend left] node {\tiny [0, 20]} (E)
	(C) edge [bend right] node [swap] {\tiny [0, 20]} (F)
	(D) edge node [swap] {\tiny [0, 5]} (E)
	(D) edge [bend right] node {\tiny [0, 5]} (F)
	(E) edge [bend left] node {\tiny [80, 100]} (F)
	(F) edge [bend left] node [swap] {\tiny [0, 20]} (E)
	(E) edge [bend left] node {\tiny [0, 20]} (G)
	(F) edge [bend right] node [swap] {\tiny [80, 100]} (G)
	(G) edge [bend left] node [swap] {\tiny [80, 80]} (H)
	(G) edge node {\tiny [15, 15]} (I)
	(G) edge [bend right] node {\tiny [0, 5]} (J)
	(K) edge node {\tiny [100, 100]} (A)
	(H) edge node {\tiny [80, 80]} (L)
	(I) edge node {\tiny [15, 15]} (M)
	(J) edge node {\tiny [0, 5]} (N)
	;
	\end{tikzpicture}
	}
	\caption{The \emph{Uncertain Directly-Follows Graph} (UDFG) computed based on the uncertain event log $\langle a, b, e, f, g, h \rangle^{80}$, $\langle a, \{b, c\}, [e, f], g, i \rangle^{15}$, $\langle a, \{b, c, d\}, [e, f], g, \overline{j} \rangle^{5}$. The arcs are labeled with the minimum and maximum number of directly-follows relationship observable between activities in the corresponding trace. Notice the large amount of connections extracted from a single and rather short trace. Uncertain directly-follows relationships are inferred from the behavior graphs of the traces in the log. The construction of this object is necessary to perform automatic process discovery over uncertain event data.}
	\label{fig:udfg}
\end{figure}

\begin{figure}[]
	\centering
	\subcaptionbox{A process model that can only replay the relationships appearing in the certain parts of the traces in the uncertain log. Here, information from uncertainty has been excluded completely.\label{subfig:1}\\}{%
		\begin{tikzpicture}[node distance=.7cm and .3cm, >=stealth']
		
		\tikzstyle{place} = [circle,draw,thick,minimum size=6mm]
		\tikzstyle{transition} = [rectangle,draw,thick,minimum size=4mm]
		\tikzstyle{invisible} = [transition, fill=black]
		\tikzstyle{finaltoken} = [token, fill=black!30]
		
		\node [place,tokens=1] (p1) {};
		\node [transition] (A) [right= of p1] {a};
		\node [place] (p2) [right= of A] {};
		\node [transition] (B) [right= of p2] {b};
		\node [place] (p3) [right= of B] {};
		\node [transition] (E) [right= of p3] {e};
		\node [place] (p4) [right= of E] {};
		\node [transition] (F) [right= of p4] {f};
		\node [place] (p5) [right= of F] {};
		\node [transition] (G) [right= of p5] {g};
		\node [place] (p6) [right= of G] {};
		\node [transition] (H) [above right= of p6] {h};
		\node [transition] (I) [below right= of p6] {i};
		\node [place] (p7) [below right= of H] {};
		\node [finaltoken] at (p7) {};
		
		\draw [->] (p1) to (A.west);
		\draw [->] (A.east) to (p2);
		\draw [->] (p2) to (B.west);
		\draw [->] (B.east) to (p3);
		\draw [->] (p3) to (E.west);
		\draw [->] (E.east) to (p4);
		\draw [->] (p4) to (F.west);
		\draw [->] (F.east) to (p5);
		\draw [->] (p5) to (G.west);
		\draw [->] (G.east) to (p6);
		\draw [->] (p6) to (H.west);
		\draw [->] (p6) to (I.west);
		\draw [->] (H.east) to (p7);
		\draw [->] (I.east) to (p7);
		\end{tikzpicture}
	}%
	\hfill
	\subcaptionbox{A process model that can replay some -- but not all -- the relationships appearing in the uncertain parts of the traces in the uncertain log. This process model mediates between representing only certain observation and representing all the possible behavior in the process.\label{subfig:2}\\}{%
		\begin{tikzpicture}[node distance=.7cm and .3cm, >=stealth']
		
		\tikzstyle{place} = [circle,draw,thick,minimum size=6mm]
		\tikzstyle{transition} = [rectangle,draw,thick,minimum size=4mm]
		\tikzstyle{invisible} = [transition, fill=black]
		\tikzstyle{finaltoken} = [token, fill=black!30]
		
		\node [place,tokens=1] (p1) {};
		\node [transition] (A) [right= of p1] {a};
		\node [place] (p2) [right= of A] {};
		\node [transition] (B) [above right= of p2] {b};
		\node [transition] (C) [below right= of p2] {c};
		\node [place] (p22) [below right= of B] {};
		\node [invisible] (Z) [right= of p22] {};
		\node [place] (p3) [above right= of Z] {};
		\node [place] (p31) [below right= of Z] {};
		\node [transition] (E) [right= of p3] {e};
		\node [transition] (F) [right= of p31] {f};
		\node [place] (p51) [right= of E] {};
		\node [place] (p5) [right= of F] {};
		\node [transition] (G) [above right= of p5] {g};
		\node [place] (p6) [right= of G] {};
		\node [transition] (H) [above right= of p6] {h};
		\node [transition] (I) [below right= of p6] {i};
		\node [place] (p7) [below right= of H] {};
		\node [finaltoken] at (p7) {};
		
		\draw [->] (p1) to (A.west);
		\draw [->] (A.east) to (p2);
		\draw [->] (p2) to (B.west);
		\draw [->] (p2) to (C.west);
		\draw [->] (B.east) to (p22);
		\draw [->] (C.east) to (p22);
		\draw [->] (p22) to (Z.west);
		\draw [->] (Z.east) to (p3);
		\draw [->] (Z.east) to (p31);
		\draw [->] (p3) to (E.west);
		\draw [->] (p31) to (F.west);
		\draw [->] (E.east) to (p51);
		\draw [->] (F.east) to (p5);
		\draw [->] (p51) to (G.west);
		\draw [->] (p5) to (G.west);
		\draw [->] (G.east) to (p6);
		\draw [->] (p6) to (H.west);
		\draw [->] (p6) to (I.west);
		\draw [->] (H.east) to (p7);
		\draw [->] (I.east) to (p7);
		\end{tikzpicture}
	}%
	\hfill
	\subcaptionbox{A process model that can replay all possible configurations of certain and uncertain traces in the uncertain log. This process model has the highest possible replay fitness, but is also very likely to contain some noisy or otherwise unwanted behavior.\label{subfig:3}\\}{%
		\begin{tikzpicture}[node distance=.7cm and .3cm, >=stealth']
		
		\tikzstyle{place} = [circle,draw,thick,minimum size=6mm]
		\tikzstyle{transition} = [rectangle,draw,thick,minimum size=4mm]
		\tikzstyle{invisible} = [transition, fill=black]
		\tikzstyle{finaltoken} = [token, fill=black!30]
		
		\node [place,tokens=1] (p1) {};
		\node [transition] (A) [right= of p1] {a};
		\node [place] (p2) [right= of A] {};
		\node [transition] (B) [above right= of p2] {b};
		\node [transition] (D) [ below right= of p2] {d};
		\node [place] (p22) [below right= of B] {};
		\node [transition] (C)  at ($(p2)!0.5!(p22)$) {c};
		\node [invisible] (Z) [right= of p22] {};
		\node [place] (p3) [above right= of Z] {};
		\node [place] (p31) [below right= of Z] {};
		\node [transition] (E) [right= of p3] {e};
		\node [transition] (F) [right= of p31] {f};
		\node [place] (p51) [right= of E] {};
		\node [place] (p5) [right= of F] {};
		\node [transition] (G) [above right= of p5] {g};
		\node [place] (p6) [right= of G] {};
		
		\node [transition] (I) [above right= of p6] {i};
		\node [transition] (J) [below right= of p6] {j};
		\node [transition] (H) [above= of I] {h};
		\node [invisible] (K) [below= of J] {k};
		\node [place] (p7) [below right= of I] {};
		\node [finaltoken] at (p7) {};
		
		\draw [->] (p1) to (A.west);
		\draw [->] (A.east) to (p2);
		\draw [->] (p2) to (B.west);
		\draw [->] (p2) to (C.west);
		\draw [->] (p2) to (D.west);
		\draw [->] (B.east) to (p22);
		\draw [->] (C.east) to (p22);
		\draw [->] (D.east) to (p22);
		\draw [->] (p22) to (Z.west);
		\draw [->] (Z.east) to (p3);
		\draw [->] (Z.east) to (p31);
		\draw [->] (p3) to (E.west);
		\draw [->] (p31) to (F.west);
		\draw [->] (E.east) to (p51);
		\draw [->] (F.east) to (p5);
		\draw [->] (p51) to (G.west);
		\draw [->] (p5) to (G.west);
		\draw [->] (G.east) to (p6);
		\draw [->] (p6) to (H.west);
		\draw [->] (p6) to (I.west);
		\draw [->] (p6) to (J.west);
		\draw [->] (p6) to (K.west);
		\draw [->] (H.east) to (p7);
		\draw [->] (I.east) to (p7);
		\draw [->] (J.east) to (p7);
		\draw [->] (K.east) to (p7);
		\end{tikzpicture}
	}%
	\caption{Three different process models for the uncertain event log $\langle a, b, e, f, g, h \rangle^{80}$, $\langle a, \{b, c\}, [e, f], g, i \rangle^{15}$, $\langle a, \{b, c, d\}, [e, f], g, \overline{j} \rangle^{5}$ obtained through inductive mining over an uncertain directly-follows graph. The different filtering parameters for the UDFG yield models with distinct features.}
	\label{fig:uncertaindiscovery}
\end{figure}

When two events have mutually overlapping timestamps, we write their activity labels between square brackets, and we indicate indeterminate events by overlining them\footnote{Notice that this notation does not allow for the representation of every possible uncertain trace: in the case of timestamp uncertainty, it can only express mutual overlapping of time intervals. However, this notation is adequate to illustrate an example for process discovery under uncertainty.}. For instance, the trace $\langle \overline{a}, \{b, c\}, [d, e] \rangle$ is a trace containing 4 events, of which the first is an indeterminate event with activity label $a$, the second is an uncertain event that can have either $b$ or $c$ as activity label, and the last two events have an interval as timestamp (and the two ranges overlap). Let us consider the following event log:

$\langle a, b, e, f, g, h \rangle^{80}, \langle a, \{b, c\}, [e, f], g, i \rangle^{15}, \langle a, \{b, c, d\}, [e, f], g, \overline{j} \rangle^{5}$.

For each pair of activities, we can count the minimum and maximum occurrences of a directly-follows relationship that can be observed in the log. The resulting UDFG is shown in Figure~\ref{fig:udfg}.

This graph can be then utilized to discover process models of uncertain logs via process discovery methods based on directly-follows relationships. In a previous work~\cite{pegoraro2019discovering} we illustrated this principle by applying it to the Inductive Miner, a popular discovery algorithm~\cite{leemans2013discovering}; the edges of the UDFG can be filtered via the information on the labels, in such a way that the final model can represent all possible behavior in the uncertain log, or only a part. Figure~\ref{fig:uncertaindiscovery} shows some process models obtained through inductive mining of the UDFG, as well as a description regarding how the model relate to the original uncertain log.

UDFGs of uncertain event data are obtained on the basis of the behavior graphs of the traces in an uncertain event log, making their construction a necessary step to perform uncertain process discovery. In fact, the frequency information labeling the edges of UDFGs are obtained through a search among the possible connections within the behavior graphs of all the traces in an uncertain log.

Thus, the construction of behavior graphs for uncertain traces is the basis of both process discovery and conformance checking on uncertain event data, since the behavior graph is a necessary processing step to mine information from uncertain traces. It is then important to be able to quickly and efficiently build the behavior graph of any given uncertain trace, in order to enable performant process discovery and conformance checking.

\section{Materials and Methods}\label{sec:materialsandmethods}

\subsection{Preliminaries}
Let us illustrate some basic concepts and notations, partially from~\cite{van2016process}:

\begin{definition}[Power set]
	The power set of a set $A$ is the set of all possible subsets of $A$, and is denoted with $\mathcal{P}(A)$. $\mathcal{P}_{NE}(A)$ denotes the set of all the non-empty subsets of $A$: $\mathcal{P}_{NE}(A) = \mathcal{P}(A)\setminus\{\emptyset\}$.
\end{definition}

\begin{definition}[Multiset]
	A \emph{multiset} is an extension of the concept of set that keeps track of the cardinality of each element. $\bag(A)$ is the set of all multisets over some set $A$. Multisets are denoted with square brackets, e.g. $b = [x, x, y]$, or with the cardinality of the elements as superscript, e.g. $b = [x^2, y]$. We denote the empty multiset with $[~]$. The operator $(\cdot)$ retrieves the cardinality of an element of the multiset, e.g. $b(x) = 2$, $b(y) = 1$, $b(z) = 0$. Over multisets we define $x \in b \Leftrightarrow b(x) \geq 1$, and $\textit{set}(b) = \{x \in b\}$. The multiset union $b = b_1 \uplus b_2$ is the multiset $b$ such that for all $x$ we have $b(x) = b_1(x) + b_2(x)$.
\end{definition}

\begin{definition}[Sequence and permutation]
	Given a set $X$, a finite \emph{sequence} over $X$ of length $n$ is a function $s \in X^* : \{1, \dots, n\} \rightarrow X$, and is written as $s = \langle s_1, s_2, \dots, s_n\rangle$. For any sequence $s$ we define $|s| = n$, $s[i] = s_i$, $x \in s \Leftrightarrow x \in \{s_1, s_2, \dots, s_n\}$ and $s \oplus s_0 = \langle s_1, s_2, \dots, s_n, s_0 \rangle$. A \emph{permutation} of the set $X$ is a sequence $x_\mathcal{S}$ that contains all elements of $X$ without duplicates: $x_\mathcal{S} \in X$, $X \in x_\mathcal{S}$, and for all $1 \leq i \leq |x_\mathcal{S}|$ and for all $1 \leq j \leq |x_\mathcal{S}|$, $x_\mathcal{S}[i] = x_\mathcal{S}[j] \rightarrow i = j$. We denote with $\mathcal{S}_X$ all such permutations of set $X$. We overload the notation for sequences: given a sequence $s = \langle s_1, s_2, \dots, s_n\rangle$, we will write $\mathcal{S}_s$ in place of $\mathcal{S}_{\{s_1, s_2, \dots, s_n\}}$.
\end{definition}

\begin{definition}[Transitive relation and correct evaluation order]\label{def:tr_rel}
	Let $X$ be a set of objects and $R$ be a binary relation $R \subseteq X \times X$. $R$ is \emph{transitive} if and only if for all $x, x', x'' \in X$ we have that $(x, x') \in R \wedge (x', x'') \in R \rightarrow (x, x'') \in R$. A \emph{correct evaluation order} is a permutation $s \in \mathcal{S}_X$ of the elements of the set $X$ such that for all $1 \leq i < j \leq |s|$ we have that $(s[i], s[j]) \in R$.
\end{definition}

\begin{definition}[Strict partial order]\label{def:st-par-ord}
	Let $S$ be a set of objects. Let $s, s' \in S$. A \emph{strict partial order} $(\prec, S)$ is a binary relation that have the following properties:
	\begin{itemize}
		\item Irreflexivity: $s \prec s$ is false.
		\item Transitivity: $s \prec s'$ and $s' \prec s''$ imply $s \prec s''$.\footnote{Formally, the third property of strict partial orders is antisimmetry: $s \prec s'$ implies that $s' \prec s$ is false. It is implied by irreflexivity and transitivity~\cite{flavska2007transitive}.}
	\end{itemize}
\end{definition}

\begin{definition}[Directed graph]
	A \emph{directed graph} $G \in \mathcal{U}_G$ is a tuple $(V, E)$ where $V$ is the set of vertices and $E \subseteq V \times V$ is the set of directed edges. The set $\mathcal{U}_G$ is the \emph{graph universe}. A \emph{path} in a directed graph $G = (V, E)$ is a sequence of vertices $p$ such that for all $1<i<|p|-1$ we have that $(p_i, p_{i+1}) \in E$. We denote with $P_G$ the set of all such possible paths over the graph G. Given two vertices $v, v' \in V$, we denote with $p_G(v, v')$ the set of all paths beginning in $v$ and ending in $v'$: $p_G(v, v') = \{p \in P_G \mid p[1] = v \wedge p[|p|] = v'\}$. $v$ and $v'$ are \emph{connected} (and $v'$ is \emph{reachable} from $v$), denoted by $v \overset{G}{\mapsto} v'$, if and only if there exists a path between them in $G$: $p_G(v, v') \neq \emptyset$. Conversely, $v \overset{G}{\not\mapsto} v' \Leftrightarrow p_G(v, v') = \emptyset$. We drop the superscript $G$ if it is clear from the context. A directed graph $G$ is \emph{acyclic} if there exists no path $p \in P_G$ satisfying $p[1] = p[|p|]$.
\end{definition}

\begin{definition}[Topological sorting]
	Let $G = (V, E)$ be an acyclic directed graph. A \emph{topological sorting~\cite{kalvin1983generation}} $o_G = \langle v_1, v_2, \dots, v_{|V|} \rangle \in \mathcal{S}_V$ is a permutation of the vertices of $G$ such that for all $1 \leq i < j \leq |V|$ we have that $v_j \not\mapsto v_i$. We denote with $\mathcal{O}_G \subseteq \mathcal{S}_V$ all such possible topological sortings over $G$.
\end{definition}

\begin{definition}[Transitive reduction]
	A \emph{transitive reduction~\cite{aho1972transitive}} $\rho \colon \mathcal{G} \to \mathcal{G}$ of a graph $G = (V, E)$ is a graph $\rho(G) = (V, E_r)$ with $E_r \subseteq E$ where every pair of vertices connected in $\rho(G)$ is not connected by any other path: for all $(v, v') \in E_r$, $p_G(v, v') = \{\langle v, v' \rangle\}$. $\rho(G)$ is the graph with the minimal number of edges that maintain the reachability between edges of $G$. The transitive reduction of a directed acyclic graph always exists and is unique\emph{~\cite{aho1972transitive}}.
\end{definition}

This paper proposes an analysis technique on \emph{uncertain event logs}. These execution logs contain information about uncertainty explicitly associated with event data. A taxonomy of different types of uncertain event logs and attribute uncertainty has been described in~\cite{pegoraro2019mining}; we will refer to the notion of \emph{simple uncertainty}, which includes uncertainty without probabilistic information on the control-flow perspective: activities, timestamps, and indeterminate events.

\begin{definition}[Universes]
	Let $\mathcal{U}_I$ be the set of all the \emph{event identifiers}. Let $\mathcal{U}_C$ be the set of all \emph{case ID identifiers}. Let $\mathcal{U}_A$ be the set of all the \emph{activity identifiers}. Let $\mathcal{U}_T$ be the totally ordered set of all the \emph{timestamp identifiers}. Let $\mathcal{U}_O = \{!, ?\}$, where the ``!'' symbol denotes \emph{determinate events}, and the ``?'' symbol denotes \emph{indeterminate events}.
\end{definition}

\begin{definition}[Simple uncertain events]
	$e = (e_i, A, t_{\text{min}}, t_{\text{max}}, o)$ is a simple uncertain event, where $e_i \in \mathcal{U}_E$ is its event identifier, $A \in \mathcal{P}_{NE}(\mathcal{U}_A$ is the set of possible activity labels for $e$, $t_{\text{min}}$ and $t_{\text{max}}$ are the lower and upper bounds for the value of its timestamp, and $o$ indicates if it is an indeterminate event. Let $\mathcal{U}_E = (\mathcal{U}_I \times \mathcal{P}_{NE}(\mathcal{U}_A) \times \mathcal{U}_T \times \mathcal{U}_T \times \mathcal{U}_O)$ be the set of all simple uncertain events. Over the uncertain event $e = (e_i, A, t_{\text{min}}, t_{\text{max}}, o)$ we define the projection functions $\pi_a(e) = A$, $\pi_{t_{\text{min}}}(e) = t_{\text{min}}$, $\pi_{t_{\text{max}}}(e) = t_{\text{max}}$ and $\pi_o(e) = o$.
\end{definition}

\begin{definition}[Simple uncertain traces and logs]
	$\sigma \subseteq \mathcal{U}_E$ is a \emph{simple uncertain trace} if for any $(e_i, A, t_{\text{min}}, t_{\text{max}}, o) \in \sigma$, $t_{\text{min}} < t_{\text{max}}$ and all the event identifiers are unique. $\mathcal{T}_U$ denotes the universe of simple uncertain traces. $L \subseteq \mathcal{T}_U$ is a \emph{simple uncertain log} if all the event identifiers in the log are unique. 
\end{definition}

\begin{definition}[Strict partial order over simple uncertain events]\label{def:ord}
	Let $e, e' \in \mathcal{E}_U^S$ be two simple uncertain events. $(\prec, \mathcal{E}_U^S)$ is an order defined on the universe of strongly uncertain events $\mathcal{E}_U^S$ as:
	\[
	e \prec e' \Leftrightarrow \pi_{t_{max}}(e) < \pi_{t_{min}}(e')
	\]
\end{definition}

\begin{definition}[Order-realizations of simple uncertain traces]\label{def:real}
	Let $\sigma \in \mathcal{T}_U$ be a simple uncertain trace. An \emph{order-realization} $\sigma_O = \langle e_1, e_2, \dots, e_{|\sigma|} \rangle \in \mathcal{S}_\sigma$ is a permutation of the events in $\sigma$ such that for all $1 \leq i < j \leq |\sigma|$ we have that $e_j \nprec e_i$, i.e. $\sigma_O$ is a correct evaluation order for $\sigma$ over $(\prec, \mathcal{E}_U^S)$, and the (total) order in which events are sorted in $\sigma_O$ is a linear extension of the strict partial order $(\prec, \mathcal{E}_U^S)$. We denote with $\mathcal{R}_O(\sigma)$ the set of all such order-realizations of the trace $\sigma$.
\end{definition}

A necessary step to allow for analysis of simple uncertain traces is to obtain their \emph{behavior graph}. A behavior graph is a directed acyclic graph that synthesizes the information regarding the uncertainty on timestamps contained in the trace.

\begin{definition}[Behavior graph]\label{def:bg}
	Let $\sigma \in \mathcal{T}_U$ be a simple uncertain trace. Let the identification function $id \colon \sigma \to \{1, 2, \dots, |\sigma|\}$ be a bijection between the events in $\sigma$ and the first $|\sigma|$ natural numbers. A behavior graph $\beta \colon \mathcal{T}_U \to \mathcal{U}_G$ is the transitive reduction of a directed graph $\rho(G)$, where $G = (V, E) \in \mathcal{U}_G$ is defined~\footnote{A technical note: this definition for the nodes of the behavior graph is slightly different from the one in~\cite{pegoraro2019mining}, to simplify the notation in algorithms. The two definitions are functionally identical. } as:
	\begin{itemize}
		\item $V = \{(\text{id}(e), \pi_a(e), \pi_o(e)) \mid e \in \sigma \}$
		\item $E = \{(v, w) \mid v, w \in V \wedge \pi_{t_{\text{max}}}(v) < \pi_{t_{\text{min}}}(w)\}$
	\end{itemize}
	The set of topological sortings of a behavior graph $\beta(\sigma)$ corresponds to the set of all the order-realizations of the trace $\sigma$: 
\end{definition}

Figures~\ref{fig:graphcomp} and~\ref{fig:graphred} show the transitive reduction operation on the running example.

\begin{figure}
	\centering
	\begin{minipage}[t]{0.48\textwidth}
		\centering
		\begin{tikzpicture}[->, node distance=3.5cm, nodes={draw, ellipse, scale=.65}]
		
		\node[dashed]	(A)	[label=below:$e_1$]								{$\text{NightSweats}$};
		\node			(B)	[right of=A, label=below:$e_2$]					{$\{\text{PrTP, SecTP}\}$};
		\node			(C)	[below of=B, yshift=1.5cm, label=below:$e_3$]	{$\text{Splenomeg}$};
		\node			(D)	[right of=C, yshift=1cm, label=below:$e_4$]		{$\text{Adm}$};
		
		\path
		(A) edge (B)
		edge [bend left] (D)
		(B) edge (D)
		(C) edge (D);
		\end{tikzpicture}
		\caption{The behavior graph of the trace in Table~\ref{table:uncertaintrace} before applying the transitive reduction. All the nodes in the graph are pairwise connected based on precedence relationships; pairs of nodes for which the order is unknown are not connected.}
		\label{fig:graphcomp}
	\end{minipage}\hfill
	\begin{minipage}[t]{0.48\textwidth}
		\centering
		\begin{tikzpicture}[->, node distance=3.5cm, nodes={draw, ellipse, scale=.65}]
		
		\node[dashed]	(A)	[label=below:$e_1$]								{$\text{NightSweats}$};
		\node			(B)	[right of=A, label=below:$e_2$]					{$\{\text{PrTP, SecTP}\}$};
		\node			(C)	[below of=B, yshift=1.5cm, label=below:$e_3$]	{$\text{Splenomeg}$};
		\node			(D)	[right of=C, yshift=1cm, label=below:$e_4$]		{$\text{Adm}$};
		
		\path
		(A) edge (B)
		(B) edge (D)
		(C) edge (D);
		\end{tikzpicture}
		\caption{The same behavior graph after the transitive reduction. The arc between $e_1$ and $e_4$ is removed, since they are reachable through $e_2$. This graph has a minimal number of arcs while conserving the same reachability relationship between nodes.}
		\label{fig:graphred}
	\end{minipage}
\end{figure}

The semantics of a behavior graph are able to efficaciously communicate time and order information concerning the time relationships among events in the corresponding uncertain trace in a compact manner. For a behavior graph $\beta(\sigma) = (V, E)$ and two events $e_1 \in \sigma$, $e_2 \in \sigma$, $(e_1, e_2) \in E$ holds if and only if $e_1$ is immediately followed by $e_2$ for some possible values of the timestamps of the events in the trace. A consequence of this fact is that if a pair of events in the graph are unreachable, they might have occurred in any order.

Definition~\ref{def:bg} is meaningful and clear from a theoretical point of view. It rigorously defines a behavior graph and the semantics of its parts. While helpful to understand the function of behavior graphs, obtaining them from process traces following this definition -- that is, utilizing the transitive reduction -- is inefficient and slow. This hinders the analysis of logs with a large number of events, and with longer traces. It is nonetheless possible to build behavior graphs from process traces in a faster and more efficient way.

\subsection{Efficient Construction of Behavior Graphs}\label{sec:bg}
The set of steps to efficiently create a behavior graph from an uncertain trace is separated into two distinct phases, described by Algorithms~\ref{alg:list} and~\ref{alg:bg}. An uncertain event $e$ is associated with a time interval which is determined by two values: minimum and maximum timestamp of that event $\pi_{t_{\text{min}}}(e)$ and $\pi_{t_{\text{max}}}(e)$. If an event $e$ has a certain timestamp, we have that $\pi_{t_{\text{min}}}(e) = \pi_{t_{\text{max}}}(e)$.

\begin{algorithm}[]
	\caption{\textsc{TimestampList($\sigma$)}}
	\label{alg:list}
	\SetKwInOut{Input}{Input~}
	\SetKwInOut{Output}{Output~}
	\Input{~An uncertain trace $\sigma$.}
	\Output{~The list of timestamps $\mathcal{L}$ of $\sigma$.}
	
	$\mathcal{L}^* \gets \langle~\rangle$ \tcp*{Support list}
	
	$\mathcal{L} \gets \langle~\rangle$ \tcp*{List of event attributes}
	
	$\mathbb{E} \gets \textsc{Sort}(\sigma)$  \tcp*{Sorts uncertain events by minimum timestamp}
	
	$i \gets 1$
	
	\While{$i \leq |\mathbb{E}|$}{
		$\mathcal{L}^* \gets \mathcal{L}^* \oplus (\pi_{t_{\text{min}}}(e), i, e, \text{'MIN'})$
		
		$\mathcal{L}^* \gets \mathcal{L}^* \oplus (\pi_{t_{\text{max}}}(e), i, e, \text{'MAX'})$
		
		$i \gets i + 1$
	}
	
	$\textsc{Sort}(\mathcal{L}^*)$ \tcp*{Sorts the list based on timestamp value}
	
	$i \gets 1$
	
	\While{$i \leq |\mathcal{L}^*|$}{
		$(t, id, e, \textit{type}) \gets \mathcal{L}^*[i]$
		
		$\mathcal{L} \gets \mathcal{L} \oplus (id, \pi_a(e), \pi_o(e), \textit{type})$
		
		$i \gets i + 1$
	}
	
	\Return $\mathcal{L}$
\end{algorithm}

\begin{algorithm}[]
	\caption{\textsc{BehaviorGraph(TimestampList($\sigma$))}}
	\label{alg:bg}
	\SetKwInOut{Input}{Input~}
	\SetKwInOut{Output}{Output~}
	\Input{~The list $\mathcal{L} = \textsc{TimestampList}(\sigma)$ of an uncertain trace $\sigma$.}
	\Output{~The behavior graph $\beta(\sigma) = (V, E)$.}
	
	$V \gets \{(id, \pi_a(e), \pi_o(e)) \mid (id, \pi_a(e), \pi_o(e), \textit{type}) \in \mathcal{L} \}$ 
	
	$E \gets \varnothing$ 
	
	$i \gets 1$
	
	\While{$i < |\mathcal{L}|$}{
		$(id, a, o, \textit{type}) \gets \mathcal{L}[i]$
		
		\If{$\textit{type} = \text{'MAX'}$}{
			$j \gets i + 1$
			
			\While{$j \leq |\mathcal{L}|$}{
				$(id^*, a^*, o^*, \textit{type}^*) \gets \mathcal{L}[j]$
				
				\uIf {$\textit{type}^* = \text{'MIN'}$}{
					$E \gets E \cup \{((id, a, o), (id^*, a^*, o^*))\}$
				}
				\ElseIf {$((id, a, o), (id^*, a^*, o^*)) \in E$}{
					\Break
				}
				$j \gets j + 1$
			}
		}
		$i \gets i + 1$
	}
	\Return $(V, E)$
\end{algorithm}

\begin{table}[]
	\caption{Running example for the creation of the behavior graph.}
	\label{table:running}
	\centering
	\begin{tabular}{ccccc}
		\textbf{Case ID} & \textbf{Event ID}        & \textbf{Activity}                                                                                                     & \textbf{Timestamp}             & \multicolumn{1}{l}{\textbf{Event Type}} \\ \hline
		\multicolumn{1}{|c|}{872} & \multicolumn{1}{|c|}{$e_1$} &
		\multicolumn{1}{c|}{a} & \multicolumn{1}{c|}{\begin{tabular}[c]{@{}c@{}} 05-12-2011\end{tabular}}                                                                                 & \multicolumn{1}{c|}{!}                    \\ \hline
		\multicolumn{1}{|c|}{872} & \multicolumn{1}{|c|}{$e_2$} &
		\multicolumn{1}{c|}{b} & \multicolumn{1}{c|}{\begin{tabular}[c]{@{}c@{}}[06-12-2011, 10-12-2011]\end{tabular}}                                                                         &  \multicolumn{1}{c|}{!}                    \\ \hline
		\multicolumn{1}{|c|}{872} & \multicolumn{1}{|c|}{$e_3$} &
		\multicolumn{1}{c|}{c}        & \multicolumn{1}{c|}{07-12-2011} &  \multicolumn{1}{c|}{!}                    \\ \hline
		\multicolumn{1}{|c|}{872} & \multicolumn{1}{|c|}{$e_4$} &
		\multicolumn{1}{c|}{d} & \multicolumn{1}{c|}{\begin{tabular}[c]{@{}c@{}}[08-12-2011, 11-12-2011]\end{tabular}}                                                                         &  \multicolumn{1}{c|}{!}                    \\ \hline
		\multicolumn{1}{|c|}{872} & \multicolumn{1}{|c|}{$e_5$} &
		\multicolumn{1}{c|}{e}        & \multicolumn{1}{c|}{\begin{tabular}[c]{@{}c@{}}09-12-2011 \end{tabular}}                                                                         &  \multicolumn{1}{c|}{!}                    \\ \hline
		\multicolumn{1}{|c|}{872} & \multicolumn{1}{|c|}{$e_6$} &
		\multicolumn{1}{c|}{f}        & \multicolumn{1}{c|}{[12-12-2011, 13-12-2011]}                                                                         &  \multicolumn{1}{c|}{!}                    \\ \hline
	\end{tabular}
\end{table}

We will examine here the effect of Algorithms~\ref{alg:list} and~\ref{alg:bg} on a running example, the process trace shown in Table~\ref{table:running}. Notice that, in this running example, no uncertainty on activity labels nor indeterminate events are present: this is because of the fact that the topology of a behavior graph only depends on the (uncertain) timestamps in the events belonging to the corresponding trace.

The construction of the graph relies on a preprocessing step shown in Algorithm~\ref{alg:list}, where a support list $\mathcal{L}$ is created (lines 4-8). Every entry in this list is a tuple of four elements. For each event $e$ in the trace, we insert two entries in the list -- one for each timestamp $\pi_{t_{\text{min}}}$ and $\pi_{t_{\text{max}}}$ appearing in a trace. The four elements in each tuple contained in the list are:
\begin{itemize}
	\item an \emph{identifier}, which in the list construction is an integer representing the rank of the uncertain event by minimum timestamp (computed in line 3);
	\item the activity labels associated with the event $\pi_a(e)$;
	\item the attribute $\pi_o(e)$, which will carry the information regarding indeterminate events;
	\item the type of timestamp that generated this entry -- if it is a minimum or maximum of an interval.
\end{itemize}
As we can see, the list is designed to contain all information about an uncertain event except the values of minimum and maximum timestamps, which we use to sort the list (line 9) and then discard prior to returning the list (lines 10-15).

\begin{table}[]
	\caption{Entries for the list $\mathcal{L}$ generated by each event in the uncertain trace. Every event $e$ has two associated entries, one marked as $\text{'MIN'}$ and the other as $\text{'MAX'}$. Each entry is a 4-uple containing an integer that acts as event identifier, the set of possible activity labels $\pi_a(e)$ of the uncertain event, the indeterminate event attribute $\pi_o(e)$, and the type of timestamp ($\text{'MIN'}$ or $\text{'MAX'}$).}
	\label{table:entries}
	\centering
	\begin{tabular}{ccc}
		\multicolumn{1}{c}{\textbf{Event}}  & \multicolumn{1}{c}{\begin{tabular}[c]{@{}c@{}}\textbf{List $\mathcal{L}^*$ entry}\\ \textbf{(minimum timestamp)}\end{tabular}} & \multicolumn{1}{c}{\begin{tabular}[c]{@{}c@{}}\textbf{List $\mathcal{L}^*$ entry}\\ \textbf{(maximum timestamp)}\end{tabular}} \\ \hline
		\multicolumn{1}{|c|}{$e_1$} & \multicolumn{1}{c|}{(05-12-2011, 1, $\{a\}$, !, \text{'MIN'})}                                                                       & \multicolumn{1}{c|}{(05-12-2011, 1, $\{a\}$, !, \text{'MAX'})}                                                                       \\ \hline
		\multicolumn{1}{|c|}{$e_2$} & \multicolumn{1}{c|}{(06-12-2011, 2, $\{b\}$, !, \text{'MIN'})}                                                                       & \multicolumn{1}{c|}{(10-12-2011, 2, $\{b\}$, !, \text{'MAX'})}                                                                       \\ \hline
		\multicolumn{1}{|c|}{$e_3$} & \multicolumn{1}{c|}{(07-12-2011, 3, $\{c\}$, !, \text{'MIN'})}                                                                       & \multicolumn{1}{c|}{(07-12-2011, 3, $\{c\}$, !, \text{'MAX'})}                                                                       \\ \hline
		\multicolumn{1}{|c|}{$e_4$} & \multicolumn{1}{c|}{(08-12-2011, 4, $\{d\}$, !, \text{'MIN'})}                                                                       & \multicolumn{1}{c|}{(08-12-2011, 4, $\{d\}$, !, \text{'MAX'})}                                                                       \\ \hline
		\multicolumn{1}{|c|}{$e_5$} & \multicolumn{1}{c|}{(09-12-2011, 5, $\{e\}$, !, \text{'MIN'})}                                                                       & \multicolumn{1}{c|}{(09-12-2011, 5, $\{e\}$, !, \text{'MAX'})}                                                                       \\ \hline
		\multicolumn{1}{|c|}{$e_6$} & \multicolumn{1}{c|}{(12-12-2011, 6, $\{f\}$, !, \text{'MIN'})}                                                                       & \multicolumn{1}{c|}{(13-12-2011, 6, $\{f\}$, !, \text{'MAX'})}                                                                       \\ \hline
	\end{tabular}
\end{table}

The events of the trace in Table~\ref{table:running} are represented in the list $\mathcal{L}^*$ by entries shown in Table~\ref{table:entries}. These entries are then sorted by Algorithm~\ref{alg:list} yielding the following list $\mathcal{L}$:
\[
\mathcal{L} = \langle (1, \{a\}, !, \text{'MIN'}), (1, \{a\}, !, \text{'MAX'}), (2, \{b\}, !, \text{'MIN'}), (3, \{c\}, !, \text{'MIN'}), 
\]
\[
(3, \{c\}, !, \text{'MAX'}), (4, \{d\}, !, \text{'MIN'}), (5, \{e\}, !, \text{'MIN'}), (5, \{e\}, !, \text{'MAX'}),
\]
\[
(2, \{b\}, !, \text{'MAX'}), (4, \{d\}, !, \text{'MAX'}), (6, \{f\}, !, \text{'MIN'}), (6, \{f\}, !, \text{'MAX'}) \rangle
\]

One of the purposes the list $\mathcal{L}$ serves is gathering the structural information to create the behavior graph; in fact, visiting the list in order is equivalent of sweeping the events of the trace on the time dimension, encountering each timestamp (minimum or maximum) sorted through time. We can visualize this on the Gantt diagram representation of the trace of Table~\ref{table:running}, visible in Figure~\ref{fig:bg_running_table}.

\begin{figure}
	\centering
	\includegraphics[width=.7\linewidth, keepaspectratio]{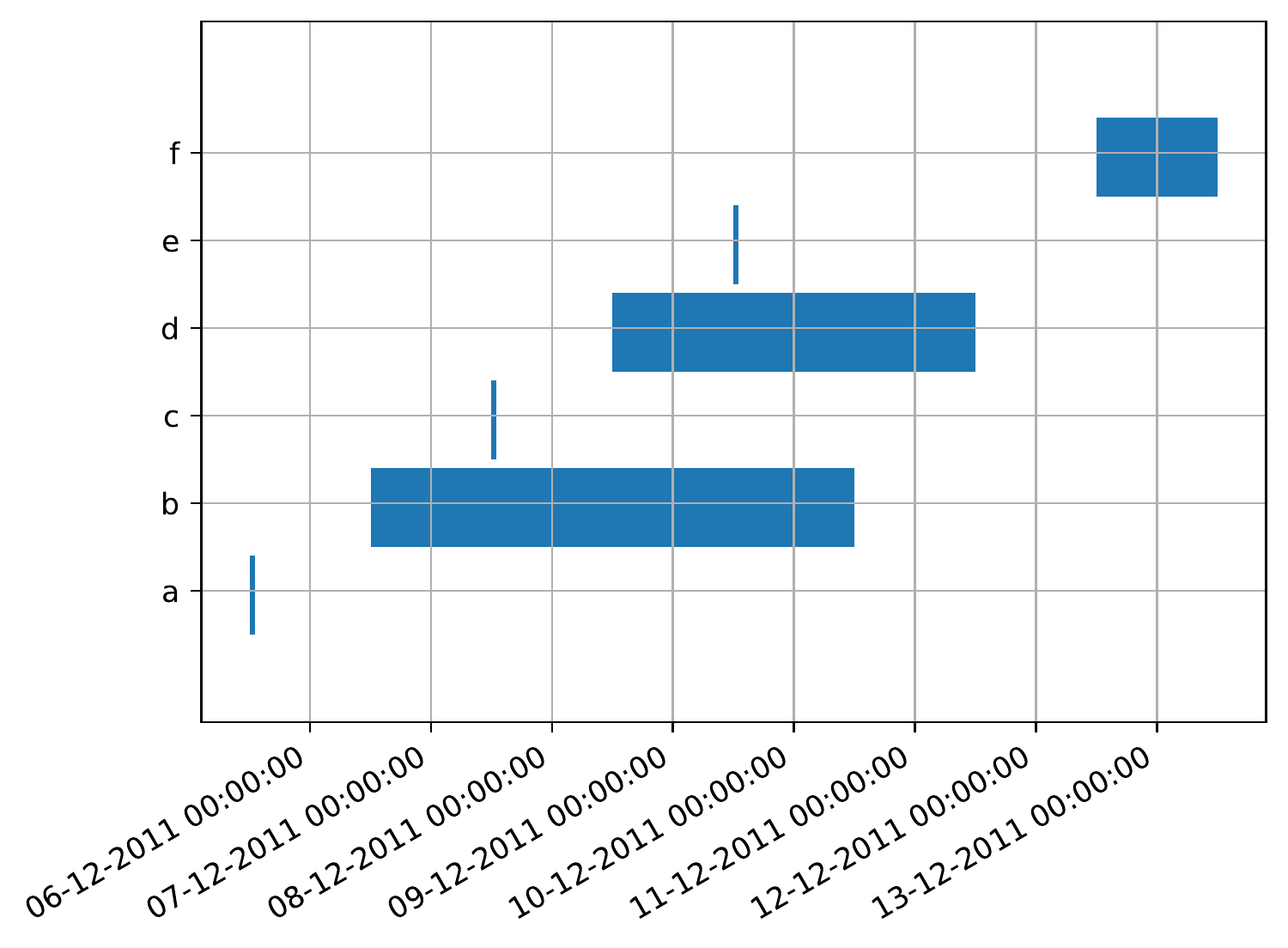}
	\caption{A Gantt diagram visualizing the time perspective of the events in Table~\ref{table:running}. The horizontal blue bars represent the interval of possible timestamps of uncertain events: such interval is ample for the event with activity label ``c'', which has an uncertain timestamp, and is narrow to indicate a precise point in time for the other events. This diagram is able to show the order relationship between events in a trace, as well as the dimensions of their interval of possible timestamps in scale.}
	\label{fig:bg_running_table}
\end{figure}

Every segment representing an uncertain event in the diagram is translated by $\textsc{TimestampList}$ into two entries in a sorted list, representing the two extremes of the segment. Events without an uncertain timestamp collapse into a single point in the diagram, and their corresponding two entries in the list are characterized by the same timestamp.

Now, let us examine Algorithm~\ref{alg:bg}. The idea leading the algorithm is to analyze the time relationship among uncertain events in a more precise manner, as opposed to adding a large number of edges to the graph and then removing them via transitive reduction. This is attained by searching all the viable successors of each event in the sorted timestamp list $\mathcal{L}$. We scan the list $\mathcal{L}$ with two nested loops, and we use the inner loop to look for successors of the entry selected by the outer loop. According to the semantics of behavior graphs, events with overlapping intervals as timestamps must not be connected by a path; thus, we draw outgoing edges from an event only when, reading the list, we arrive at a point in time in which the event has certainly occurred. This is the reason why outgoing edges are not drawn when inspecting minimum timestamps (line 6) and incoming edges are not drawn when inspecting maximum timestamps (line 10).

First, we initialize the set of nodes with all the triples $(id, \pi_a(e), \pi_o(e))$ in the entries of $\mathcal{L}$, and we initialize the edges with an empty set (lines 1-2). For each maximum timestamp that we encounter in the list, we start searching for successors in the following entries (lines 3-9), so we proceed in looking for the successors of $(id, a, o, \textit{type})$ only if $\textit{type} = \text{'MAX'}$.

If, while searching for successors of the entry $(id, a, o, \text{'MAX'})$, we encounter the entry $(id^*, a^*, o^*, \textit{type}^*)$ corresponding to a minimum timestamp ($ \textit{type}^* = \text{'MIN'}$), we connect $(id, a, o)$ and $(id^*, a^*, o^*)$ in the graph, since their timestamps do not have any possible value in common. The search for successors must continue, since it is possible that other events took place before the maximum timestamp of the event corresponding to $(id^*, a^*, o^*, \textit{type}^*)$. This configuration occurs for events $e_1$ and $e_3$ in Table~\ref{table:running}. As can be seen in Figure~\ref{fig:bg_running_table}, $e_3$ can indeed follow $e_1$, but the still undiscovered event $e_2$ is another possible successor for $e_1$.

If the entry $(id^*, a^*, o^*, \textit{type}^*)$ corresponds to a maximum timestamp (line 12), so $\textit{type}^* = \text{'MAX'}$, there are two separate situations to consider. Case 1: $(id, a, o)$ was not already connected to $(id^*, a^*, o^*)$. Then, the timestamps of the events corresponding to $(id, a, o)$ and $(id^*, a^*, o^*)$ overlap with each other -- if they did not, the two nodes would have already been connected, since we would have encountered $(id^*, a^*, o^*, \text{'MIN'})$ from $(id, a, o, \textit{'MAX'})$ before encountering $(id^*, a^*, o^*, \text{'MAX'})$. Thus, $(id, a, o)$ must not be connected to $(id^*, a^*, o^*)$ and the search must continue. Events $e_3$ and $e_4$ are an example: when the maximum timestamp of $e_4$ is encountered during the search for the successor of $e_3$, the two are not connected, so the search for a viable successor of $e_3$ has to continue. Case 2: $(id, a, o)$ and $(id^*, a^*, o^*)$ are already connected. This means that we had already encountered $(id^*, a^*, o^*, \text{'MIN'})$ during the search for the successors of $(id, a, o)$. Since the entire time interval representing the possible timestamp of the event associated with $(id^*, a^*, o^*)$ is detected after the occurrence of $(id, a, o)$, there are no further events to consider as successors of $(id, a, o)$ and the search stops (line 13). In the running example, this happens between $e_5$ and $e_6$: when searching for the successors of $e_5$, we first connect it with $e_6$ when we encounter its minimum timestamp; we then encounter its maximum timestamp, so no other successive event can be a successor for $e_5$.
This concludes the walkthrough of the procedure, which shows why Algorithms~\ref{alg:list} and~\ref{alg:bg} can be used to correctly compute the behavior graph of a trace. The behavior graph of the trace in Table~\ref{table:running} obtained through this procedure is shown in Figure~\ref{fig:bg_running}.

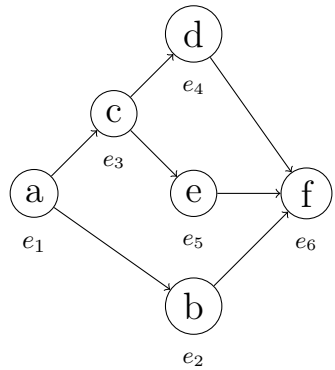
\begin{figure}
	\centering
	\begin{tikzpicture}[->, node distance=1.5cm, nodes={draw, circle}]
	
	\node	(A)	[label=below:\normalsize$e_1$]			{\Large a};
	\node	(B)	[above right of=A, label=below:$e_3$]	{\Large c};
	\node	(D)	[above right of=B, label=below:$e_4$]	{\Large d};
	\node	(E)	[below right of=B, label=below:$e_5$]	{\Large e};
	\node	(C)	[below of=E, label=below:$e_2$]			{\Large b};
	\node	(F)	[right of=E, label=below:$e_6$]			{\Large f};
	
	\path
	(A) edge (B)
	(A) edge (C)
	(B) edge (D)
	(B) edge (E)
	(D) edge (F)
	(E) edge (F)
	(C) edge (F);
	\end{tikzpicture}
	\caption{The behavior graph of the trace in Table~\ref{table:running}.}
	\label{fig:bg_running}
\end{figure}

Let us now prove, in more formal terms, the correctness of these algorithms. We will show that the procedures \textsc{BehaviorGraph} and \textsc{TimestampList} are able to construct a behavior graph with the semantics illustrated in Definition~\ref{def:bg}.

\begin{theorem}[Correctness of the behavior graph construction]
	Let $\sigma \in \mathcal{T}_U$ be an uncertain trace. Let $bg = (V, E) = \textsc{BehaviorGraph}(\textsc{TimestampList}(\sigma))$ be the behavior graph of $\sigma$ obtained through Algorithms~\ref{alg:list} and~\ref{alg:bg}. The graph $bg$ follows the behavior graph semantics: for all pairs of events $e \in \sigma$ and $e' \in \sigma$ such that $id(e) = e_{id}$, $\pi_a(e) = e_a$, $\pi_o(e) = e_o$, $id(e') = e'_{id}$, $\pi_a(e') = e'_a$, $\pi_o(e') = e'_o$, we have that the node $(e_{id}, e_a, e_o)$ is connected to the node $(e'_{id}, e'_a, e'_o)$ if and only if $\pi_{t_{\text{max}}}(e) < \pi_{t_{\text{min}}}(e')$ and there exists no event $e'' \in \sigma$ such that $\pi_{t_{\text{max}}}(e) < \pi_{t_{\text{min}}}(e'') \leq \pi_{t_{\text{max}}}(e'') < \pi_{t_{\text{min}}}(e')$. Thus, $bg = \beta(\sigma)$.
\end{theorem}
\begin{proof}
	Let us first define a suitable $id$ function for the behavior graph utilizing the list $\mathbb{E}$ created in $\textsc{TimestampList}(\sigma)$. For all events $e^* \in \sigma$ and for $i \in \mathbb{N}$ such that $\mathbb{E}[i] = e^*$, we define $id(e^*) = i$. Since $id$ is just an enumeration of the events in $\sigma$, it is trivially bijective.
	
	\leavevmode \\
	$(\Leftarrow)$
	Assume $\pi_{t_{\text{max}}}(e) < \pi_{t_{\text{min}}}(e')$. By construction, we have that\\$\mathcal{L} = \langle \dots, (e_{id}, e_a, e_o, \text{'MAX'}), \dots, (e'_{id}, e'_a, e'_o, \text{'MIN'}), \dots \rangle$. The checks in line 6 and line 10 only allow for edges to be linked from entries of type $\text{'MAX'}$ to entries of type $\text{'MIN'}$ that only appear in a later position in the list $\mathcal{L}$. Thus, the configuration $\pi_{t_{\text{max}}}(e) < \pi_{t_{\text{min}}}(e')$ is a strict prerequisite for $(e_{id}, e_a, e_o)$ and $(e'_{id}, e'_a, e'_o)$ to be connected: $((e_{id}, e_a, e_o), (e'_{id}, e'_a, e'_o)) \in E \Rightarrow \pi_{t_{\text{max}}}(e) < \pi_{t_{\text{min}}}(e')$.
	
	\leavevmode \\
	$(\Rightarrow)$
	Assume $\pi_{t_{\text{max}}}(e) < \pi_{t_{\text{min}}}(e')$, and that the algorithm is currently searching the successors for the entry $(e_{id}, e_a, e_o, \text{'MAX'})$. Eventually, the inner loop will consider as a successor the entry $(e'_{id}, e'_a, e'_o, \text{'MIN'})$, and since it is of type $\text{'MIN'}$, $(e_{id}, e_a, e_o)$ and $(e'_{id}, e'_a, e'_o)$ will necessarily be connected unless the algorithm executes the \texttt{break} at line 13. To execute it, the algorithm needs to find a list entry $(e''_{id}, e''_a, e''_o, \text{'MAX'})$ such that there already exist an arc between $(e_{id}, e_a, e_o)$ and $(e''_{id}, e''_a, e''_o)$, and this is only possible if $(e''_{id}, e''_a, e''_o, \text{'MIN'})$ has been encountered while searching for successors of $(e_{id}, e_a, e_o)$. This implies that
	\[
	\mathcal{L} = \langle \dots, (e_{id}, e_a, e_o, \text{'MAX'}), \dots, (e''_{id}, e''_a, e''_o, \text{'MIN'}), \dots
	\]
	\[
	\dots, (e''_{id}, e''_a, e''_o, \text{'MAX'}), \dots, (e'_{id}, e'_a, e'_o, \text{'MIN'}), \dots \rangle
	\]
	which, by construction of $\mathcal{L}$, is only possible if there exist some $e'' \in \sigma$ such that $$\pi_{t_{\text{max}}}(e) < \pi_{t_{\text{min}}}(e'') \leq \pi_{t_{\text{max}}}(e'') < \pi_{t_{\text{min}}}(e')$$
	\qed
\end{proof}

As mentioned earlier, the procedure of constructing a behavior graph has been structured in two different algorithms specifically to enable further optimization in processing uncertain process trace. This becomes evident once we consider the problem of converting in behavior graphs all the traces in an event log, as opposed as one single uncertain trace.

Firstly, it is important to notice that different uncertain traces can have the same list $\mathcal{L}$. Similarly to directly-follows relationships in more classical process mining, which can ignore the amount of time in absolute terms elapsed between two consecutive events, specific values of timestamps in an uncertain trace are not necessarily meaningful with respect to the connection in the behavior graph; their order, conversely, is crucial.

This fact enables further optimization at the log level. The construction of the list $\mathcal{L}$ in $\textsc{TimestampList}(\sigma)$ is engineered in a way that allows for computing the behavior graph without direct lookup to the events in the trace. This implies that it is possible to extract a multiset of lists $\mathcal{L}$ from the event log, and to compute the conversion to behavior graph only for the set of lists induced by this multiset. This allows to save computation time in converting an entire event log to behavior graphs; furthermore, it enables a more compact representation of the log in memory, since we only need to store a smaller number of graphs to represent the whole log.

The procedure to efficiently convert an event log into graphs is detailed in Algorithm~\ref{alg:log}.

\begin{algorithm}[]
	\caption{\textsc{ProcessUncertainLog}}
	\label{alg:log}
	\SetKwInOut{Input}{Input~}
	\SetKwInOut{Output}{Output~}
	\Input{~An uncertain log $L$.}
	\Output{~A multiset of behavior graphs $BG$.}
	
	$ML \gets [~]$
	
	$\mathcal{V}_L \gets [~]$
	
	\For{$\sigma \in L$}{
		$\textit{ML} \gets \textit{ML} \uplus [\textsc{TimestampList}(\sigma)]$
	}
	
	\For{$\mathcal{L} \in \textit{ML}$}{
		$\mathcal{V}_L \gets \mathcal{V}_L \uplus [\textsc{BehaviorGraph}(\mathcal{L})^{\textit{ML}(\mathcal{L})}]$
	}
	
	\Return $BG$
\end{algorithm}

These considerations allow us to extend to the uncertain scenario some concepts that are essential in classical process mining. Firstly, we can now derive the definition of \emph{variant}, highly important for preexisting process mining techniques, to uncertain event data.

\begin{definition}[Uncertain variants]
	Let $L \subseteq \mathcal{T}_U$ be a simple uncertain event log. The variants of $L$ denoted by $\mathcal{V}_L$, are the multisets of behavior graphs for the uncertain traces in $L$, and are computed with $\textsc{ProcessUncertainLog}(L)$.
\end{definition}

The computational advantage in representing a log through a multiset of behavior graphs is evident in the procedure described in Algorithm~\ref{alg:bg}. We see that all data necessary to the creation of a behavior graph is contained in the list $\mathcal{L}$, fact that justifies the log representation method illustrated in Algorithm~\ref{alg:log}.

\begin{lemma}
	Two uncertain traces $\sigma_1 \in L$ and $\sigma_2 \in L$ belong to the same variant, and share the same behavior graph, if and only if they result in the same timestamp list $\mathcal{L}$: $\textsc{TimestampList}(\sigma_1) = \textsc{TimestampList}(\sigma_2)$.
\end{lemma}

Another central concept in process mining is the so-called \emph{control-flow perspective} of event data. In certain process traces, where timestamps have a total order, events have a single activity label and no event is indeterminate, the control-flow information is represented by a sequence of activity labels sorted by timestamp. Although there are many analysis approaches that also account for other perspectives (e.g. the performance perspective, that considers the duration of events and their distance in time, or the resource perspective, that accounts for the agents that execute the activities), a vast amount of process mining techniques, including most popular algorithms for process discovery and conformance checking, rely only on the control-flow perspective of a process. Analogously, behavior graphs carry over the control-flow information of an uncertain trace: instead of describing the flow of events like their certain counterpart, the behavior graph describes all possible flows of events in the uncertain trace.

\section{Asymptotic Complexity}\label{sec:proofs}
In this section, we will provide some values for the asymptotic complexity of the algorithms seen in this paper.

In a previous paper~\cite{pegoraro2019mining} we introduced the concept of behavior graph for the representation of uncertain event data, together with a method to obtain such graphs. Definition~\ref{def:bg} describes such a baseline method for the creation of the behavior graph consisting of two main parts: the construction of the starting graph and the computation of its transitive reduction. Let us consider an uncertain process trace $\sigma \in \mathcal{T}_U$ with $|\sigma| = n$ events, and the graph $G = (V, E)$ generated in Definition~\ref{def:bg} before the transitive reduction.

The starting graph is created by inspecting the time relationship between every pair of events; this corresponds to checking if an edge exists between each pair of vertices in $G$, which needs $\mathcal{O}(n^2)$ time.

The transitive reduction of graphs can be obtained through many methods. A simple and efficient method to compute the transitive reduction on sparse graphs is to test reachability through a search (either breadth-first or depth-first) from each edge. This method costs $\mathcal{O}(V \cdot E)$ time\footnote{Here, for simplicity, we resort to a widely adopted abuse of notation in asymptotic complexity analysis: we indicate a set instead of its cardinality (e.g., we use $\mathcal{O}(V)$ in place of $\mathcal{O}(|V|)$).}. However, in the initial graph each event $e \in V$ has an inbound arc from each event certainly preceding $e$ and an outbound arc to each event certainly following $e$. Fewer events with overlapping intervals as timestamps of uncertain events imply fewer arcs in $G$; the initial graph $G$ of a trace with no uncertainty has $|E| = \frac{n(n-1)}{2} = \mathcal{O}(V^2)$ edges. Thus, except for rare, very uncertain cases, the graph $G$ is dense.

Aho et al.~\cite{aho1972transitive} presented a technique to compute the transitive reduction in $\mathcal{O}(n^3)$ time, more appropriate in the case of dense graphs, and proved that the transitive reduction has the same computational complexity of the matrix multiplication problem. The problem of matrix multiplication was generally regarded as having an optimal time complexity of $\mathcal{O}(n^3)$, until Volker Strassen presented an algorithm~\cite{strassen1969gaussian} able to multiply matrices in $\mathcal{O}(n^{2.807355})$ time. Subsequent improvements have followed, by Coppersmith and Winograd~\cite{coppersmith1990matrix}, Stothers~\cite{stothers2010complexity} and Williams~\cite{williams2012multiplying}. The asymptotically fastest algorithm known to date has been illustrated by Le Gall~\cite{le2014powers} and has an execution time of $\mathcal{O}(n^{2.3728639})$. However, these faster algorithms are very seldomly used in practice, due to the existence of large constant factors in their computation time that are hidden by the asymptotic notation. Moreover, they have vast memory requirements. The Strassen algorithm is helpful in real-life applications only when applied on very large matrices~\cite{d2005using}, and the Coppersmith-Winograd algorithm and subsequent improvements are more efficient only with inputs so large that they are effectively classified as galactic algorithms~\cite{le2012faster}.

Bearing in mind these considerations, for the vast majority of event logs, the most efficient way to implement the creation of the behavior graph via transitive reduction runs in $\mathcal{O}(n^2) + \mathcal{O}(n^3) = \mathcal{O}(n^3)$ time in the worst-case scenario.

It is straightforward to find upper bounds for the complexity of Algorithms~\ref{alg:list} and~\ref{alg:bg}.

Line 3 of $\textsc{TimestampList}$ require $\mathcal{O}(n \log n)$ to be executed. Lines 5-8 require $\mathcal{O}(n)$ time. Line 9 requires $\mathcal{O}(2n \log(2n)) = \mathcal{O}(n \log n)$ time to be run. Lines 11-14 require $2n = \mathcal{O}(n)$ time to be run. Lines 1-4 and 10 have a constant cost $\mathcal{O}(1)$. Thus, $\textsc{TimestampList}$ has a total asymptotic cost of $\mathcal{O}(1) + 2 \cdot \mathcal{O}(n \log n) + 2 \cdot \mathcal{O}(n) = \mathcal{O}(n \log n)$ in the worst-case scenario.

Let us now examine $\textsc{BehaviorGraph}$. Lines 1-3 and line 11 run in $\mathcal{O}(1)$ time. Lines 11-30 consist of two nested loops over the list $\mathcal{L}$, and we have $|\mathcal{L}| = 2n$, resulting in an asymptotic cost of $\mathcal{O}((2n)^2) = \mathcal{O}(n^2)$. The total running time for the novel construction method is then $\mathcal{O}(1) + \mathcal{O}(n^2) = \mathcal{O}(n^2)$ time in the worst-case scenario.

We can also obtain a lower bound for the complexity in the worst-case scenario by analyzing the possible size of the output. The complete directed bipartite graph with $n$ vertices, usually indicated with $K_{\frac{n}{2},\frac{n}{2}}$, is a DAG that has $(\frac{n}{4})^2 = \mathcal{O}(n^2)$ edges. It is easy to see that the complete bipartite graph fulfills the requirements to be a behavior graph: it is in fact acyclic, and no edge can be removed without changing the reachability of the graph -- namely, it is equivalent to its transitive reduction. We can show that a behavior graph with such a shape exists employing a simple construction: a trace composed by $n$ events with timestamps such that the first $\frac{n}{2}$ events all have overlapping timestamps, the last $\frac{n}{2}$ also all have overlapping timestamps, and the maximum timestamp of each of the first $\frac{n}{2}$ is smaller than the minimum timestamp of each of the last $\frac{n}{2}$ events. The construction, together with an example, is illustrated in Figure~\ref{fig:kk}. Since lines 11-30 of the algorithm build this graph with $\mathcal{O}(n^2)$ edges, the algorithm runs in $\Omega(n^2)$ time, and thus also in $\Theta(n^2)$ time. This also proves the asymptotic optimality of the algorithm: no algorithm to build behavior graphs can run in less than $\Theta(n^2)$ time in the worst-case scenario.

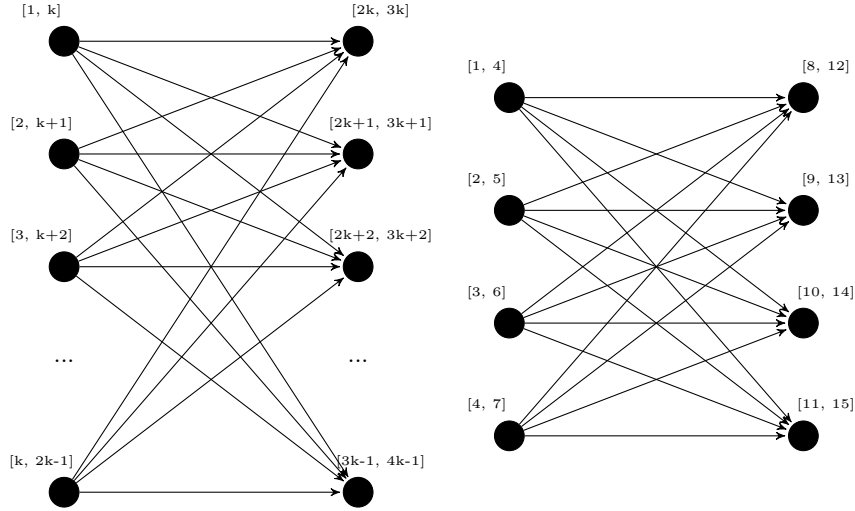
\begin{figure}[h!]
	\centering
	\begin{minipage}{.49\textwidth}
		\begin{tikzpicture}[->,>=stealth',shorten >=1pt,node distance=1.5cm,auto,main node/.style={circle,draw,align=center, fill=black}]
		
		\node[main node, label={[xshift=-0.3cm] \tiny [1, k]}]		(A)	[]				{.};
		\node[main node, label={[xshift=-0.3cm] \tiny [2, k+1]}]	(B)	[below of=A]	{.};
		\node[main node, label={[xshift=-0.3cm] \tiny [3, k+2]}]	(C)	[below of=B]	{.};
		\node[label={...}]											(D)	[below of=C]	{};
		\node[main node, label={[xshift=-0.3cm] \tiny [k, 2k-1]}]	(E)	[below of=D]	{.};
		
		\node[main node, label={[xshift=0.3cm] \tiny [2k, 3k]}]		(F)	[right=3.5cm of A]	{.};
		\node[main node, label={[xshift=0.3cm] \tiny [2k+1, 3k+1]}]	(G)	[below of=F]		{.};
		\node[main node, label={[xshift=0.3cm] \tiny [2k+2, 3k+2]}]	(H)	[below of=G]		{.};
		\node[label={...}]												(I)	[below of=H]		{};
		\node[main node, label={[xshift=0.3cm] \tiny [3k-1, 4k-1]}]	(J)	[below of=I]		{.};
		
		\path
		(A) edge (F)
		(A) edge (G)
		(A) edge (H)
		(A) edge (J)
		(B) edge (F)
		(B) edge (G)
		(B) edge (H)
		(B) edge (J)
		(C) edge (F)
		(C) edge (G)
		(C) edge (H)
		(C) edge (J)
		(E) edge (F)
		(E) edge (G)
		(E) edge (H)
		(E) edge (J)
		;
		\end{tikzpicture}
	\end{minipage}
	\begin{minipage}{.49\textwidth}
		\begin{tikzpicture}[->,>=stealth',shorten >=1pt,node distance=1.5cm,auto,main node/.style={circle,draw,align=center, fill=black}]
		
		\node[main node, label={[xshift=-0.3cm] \tiny [1, 4]}]	(A)	[]				{.};
		\node[main node, label={[xshift=-0.3cm] \tiny [2, 5]}]	(B)	[below of=A]	{.};
		\node[main node, label={[xshift=-0.3cm] \tiny [3, 6]}]	(C)	[below of=B]	{.};
		\node[main node, label={[xshift=-0.3cm] \tiny [4, 7]}]	(D)	[below of=C]	{.};
		
		\node[main node, label={[xshift=0.3cm] \tiny [8, 12]}]		(F)	[right=3.5cm of A]	{.};
		\node[main node, label={[xshift=0.3cm] \tiny [9, 13]}]		(G)	[below of=F]		{.};
		\node[main node, label={[xshift=0.3cm] \tiny [10, 14]}]	(H)	[below of=G]		{.};
		\node[main node, label={[xshift=0.3cm] \tiny [11, 15]}]	(I)	[below of=H]		{.};
		
		\path
		(A) edge (F)
		(A) edge (G)
		(A) edge (H)
		(A) edge (I)
		(B) edge (F)
		(B) edge (G)
		(B) edge (H)
		(B) edge (I)
		(C) edge (F)
		(C) edge (G)
		(C) edge (H)
		(C) edge (I)
		(D) edge (F)
		(D) edge (G)
		(D) edge (H)
		(D) edge (I)
		;
		\end{tikzpicture}
	\end{minipage}
	\caption{Construction of the class of behavior graphs isomorphic to a complete bipartite graph and an instantiated example. For any $n = 2k$, it is possible to have a behavior graph isomorphic to the graph $K_{k,k}$, which thus has a number of edges quadratic in the number of vertices.}
	\label{fig:kk}
\end{figure}

\section{Experimental Results}\label{sec:experiments}
The formal definition of our novel construction method for the behavior graph was used to show its asymptotic speedup with respect to the construction utilizing the transitive reduction. In order to empirically confirm this improvement, we built a set of experiments in order to measure the gain in speed and memory usage.

\subsection{Performance of Behavior Graph Construction}

In this section, we will show a comparison between the running time of the na{\"i}ve behavior graph construction -- which employs the transitive reduction -- versus the improved method detailed throughout the paper. The experiments are set to investigate the difference in performance between the two algorithms, and most importantly how this difference scales when the size of the event log increases, as well as the amount of events in the log that have uncertain timestamps. In designing the experiments, we took into consideration the following research questions:
\begin{itemize}
	\item \emph{Q1}: how does the computation time of the two methods compare when run on logs having an increasing number of traces?
	\item \emph{Q2}: how does the computation time of the two methods compare when run on logs with increasing trace lengths?
	\item \emph{Q3}: how does the computation time of the two methods compare when run on logs with increasing percentages of events with uncertain timestamps?
	\item \emph{Q4}: what degree of reduction in memory consumption for the representation of an uncertain log can we attain with the novel method?
	\item \emph{Q5}: do the answers obtained for \emph{Q3} hold when simulating uncertainty on real-life event data?
\end{itemize}

Both the baseline algorithm based on transitive reduction~\cite{pegoraro2019mining} and the new algorithm for the construction of the behavior graph are implemented in Python, within the PROVED project. The implementation of both methods is available online, as well as the full code for the experiments presented here (see the reference in Section~\ref{sec:introduction}).

For each series of experiments exploring \emph{Q1} through \emph{Q4}, we generate a synthetic event log with a number $n$ of traces of length $l$ (in number of events belonging to the trace). Uncertainty on timestamps is then artificially added to the events in the log. A specific percentage $p$ of the events in the event log will have an uncertain timestamp, causing it to overlap with an adjacent event. Finally, behavior graphs are built from all the traces in the event log with either algorithm, while the execution time is measured. All results in this section are presented as the mean of the measurements for 10 runs of the corresponding experiment. In the diagrams, we will label with ``TrRed'' the na{\"i}ve method using the transitive reduction, and with ``Improved'' the faster algorithm illustrated in this paper. Additionally, the data series for the novel method are labeled with the relative variation in running time for each specific data point in the experiment, expressed in percentage.

To answer \emph{Q1}, the first experiment inspects how the efficiency of the two algorithms scales with log dimension in number of traces. We generate logs with a fixed uncertainty percentage of $p = 0.5$, and trace length of $l = 20$. The number of traces in the uncertain log progressively scales from $n = 1000$ to $n = 10000$. As shown in Figure~\ref{fig:ntraces}, our proposed algorithm outperforms the baseline algorithm, showing a much smaller slope in computation time. As anticipated by the theoretical analysis, the computing time to build behavior graphs increases linearly with the number of traces in the event log for both methods; in the novel method, the constant factors are much smaller, thus producing the speedup that we can observe in the graph. Note that in this experiment the novel method requires between $18\%$ and $26\%$ of the time with respect to the baseline method.

\begin{figure}[H]
	\centering
	\includegraphics[width=.8\linewidth, keepaspectratio]{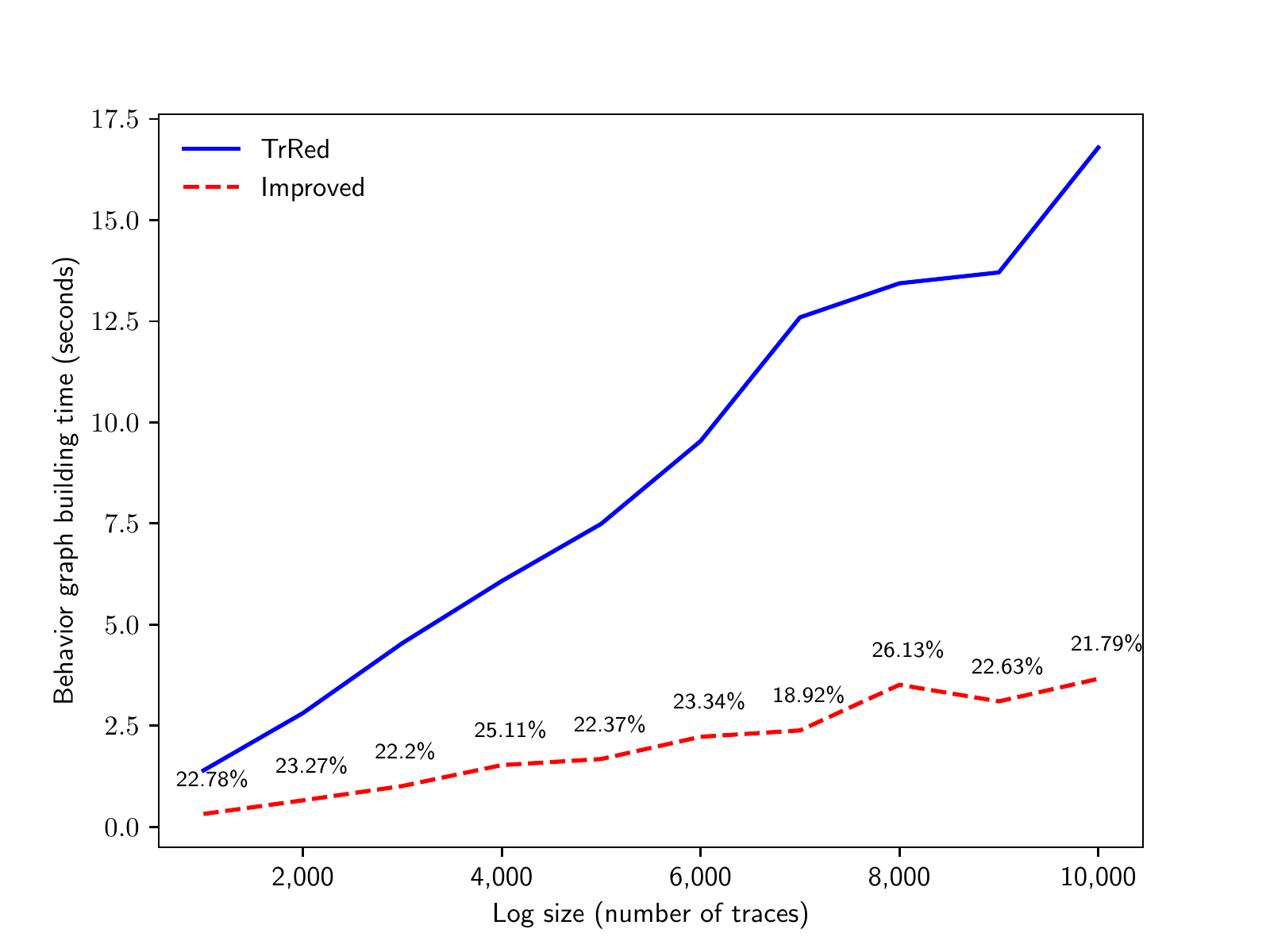}
	\caption{Time in seconds for the creation of the behavior graphs for synthetic logs with traces of length $l = 20$ events and $p = 0.5$ of uncertain events, with increasing number of traces $n$. The solid blue line indicates the time needed for the na{\"i}ve construction; the dashed red line shows the building time of the improved algorithm, and is labeled with the relative time variation (in percentage).}
	\label{fig:ntraces}
\end{figure}

The second experiment is designed to answer \emph{Q2}. We analyze the effect of the trace length on the total time needed for behavior graph creation. Therefore, we created logs with $n = 100$ traces of increasing lengths in number of events, and added uncertain timestamps to events with $p = 0.5$. The results, illustrated by Figure~\ref{fig:lengths}, meet our expectations: the computation time of the baseline method scales much worse than the computation time required by our new technique, due to its cubic asymptotic time complexity. This confirms the results of the analysis of the asymptotic time complexity analysis detailed in Section~\ref{sec:proofs}. We can notice an order-of-magnitude increase in speed. At trace length $l = 600$, the new algorithm computes the graphs in only $0.35\%$ of the time required by the baseline algorithm.

\begin{figure}[H]
	\centering
	\includegraphics[width=.8\linewidth, keepaspectratio]{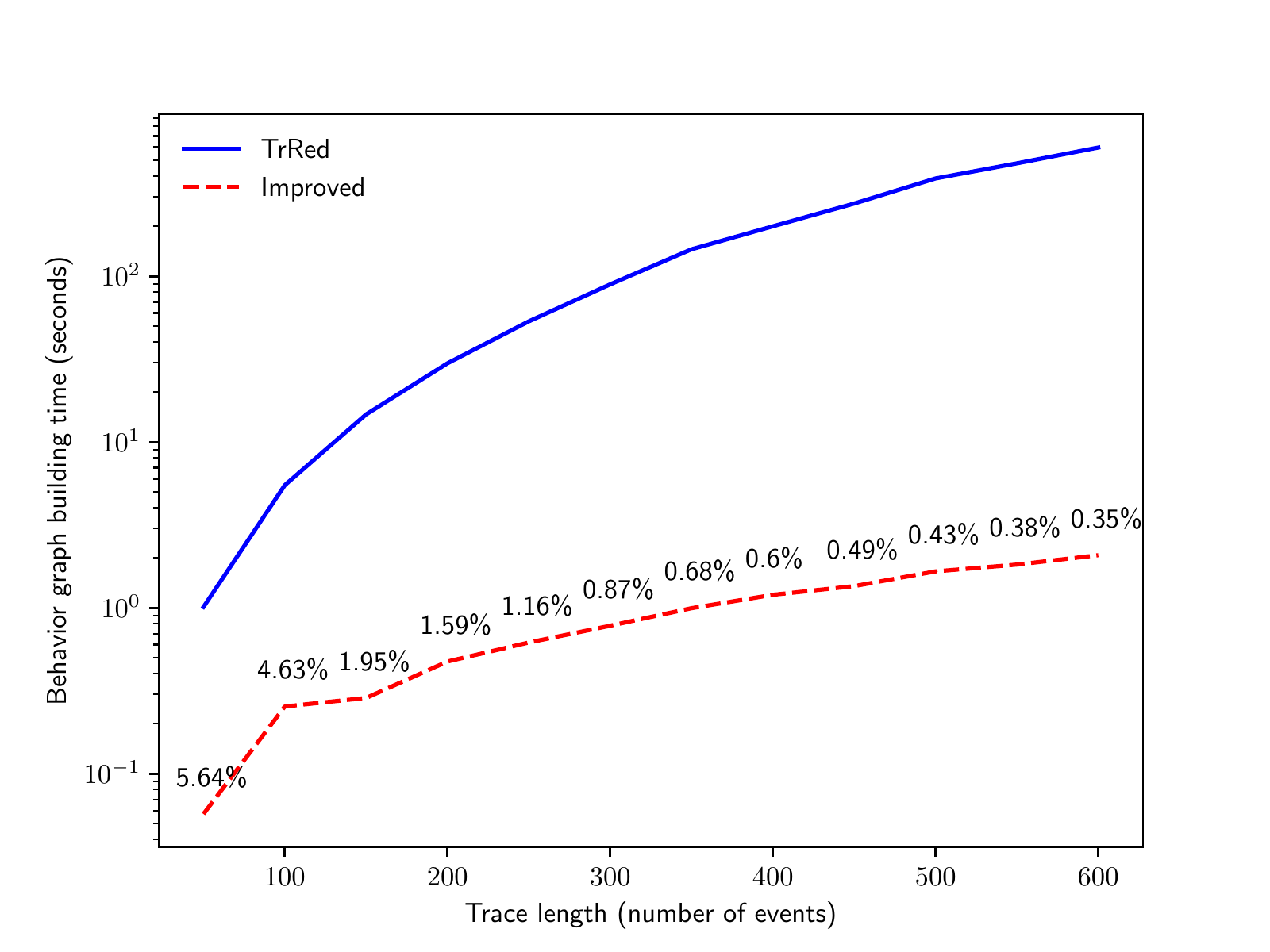}
	\caption{Time in seconds for the creation of the behavior graphs for synthetic logs with $n = 100$ traces and $p = 0.5$ of uncertain events, with increasing trace length $l$.}
	\label{fig:lengths}
\end{figure}

The next experiment tackles \emph{Q3}, by inspecting the difference in execution time for the two algorithms in function of the percentage of events with an uncertain timestamp in the event log. Keeping constant the values $n = 100$ and $l = 100$, we progressively increased the percentage $p$ of events with an uncertain timestamp and measured computation time. As presented in Figure~\ref{fig:probs}, the time required for behavior graph construction remains almost constant for our proposed algorithm, while it is very slightly decreasing for the baseline algorithm. This behavior is expected, and is justified by the fact that the worst-case scenario for the baseline algorithm is a trace that has no uncertainty on the timestamp: in that case, the behavior graph is simply a chain of nodes representing the total order in a sequence of events with certain timestamps, thus the transitive reduction needs to find and remove a higher number of edges from the directed graph. This worst-case scenario occurs at $p = 0$, explaining why the computation time needed by the transitive reduction is at its highest. It is important to note, however, that for all values of $p$ our new algorithm runs is significantly more efficient than the baseline algorithm: with $p = 0$, the new algorithm takes $0.47\%$ of the time needed by the na{\"i}ve construction, while for $p = 1$ this figure grows to $4.39\%$.

\begin{figure}[H]
	\centering
	\includegraphics[width=.8\textwidth, keepaspectratio]{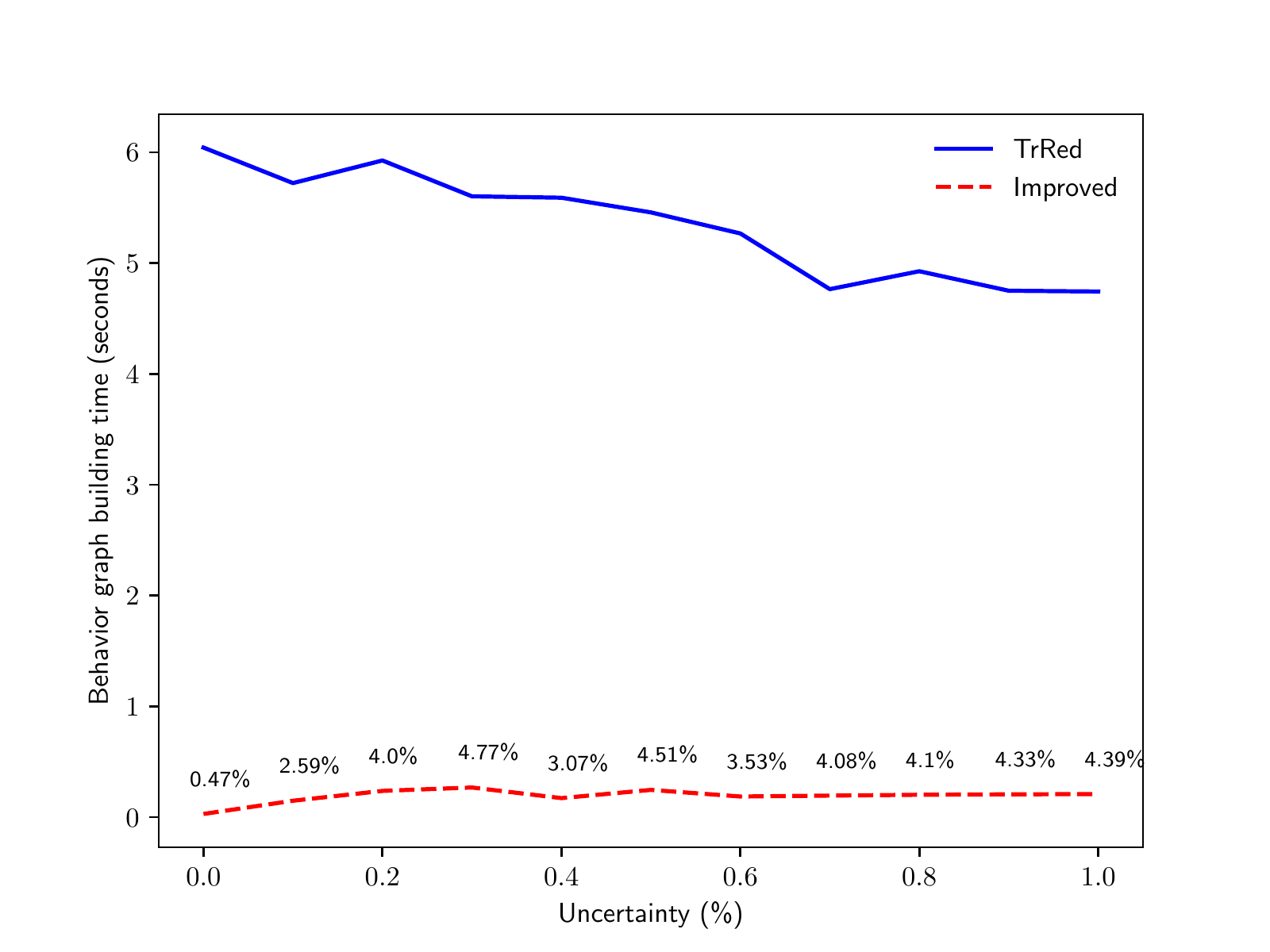}
	\caption{Time in seconds for the creation of the behavior graphs for synthetic logs with $n = 100$ traces of length $l = 100$ events, with increasing percentages of timestamp uncertainty $p$.}
	\label{fig:probs}
\end{figure}

An additional experiment is illustrated to provide an answer to \emph{Q4}. Similarly to the first experiment, we increase the number of traces $n$ in the uncertain log, while keeping the other parameters fixed: $l = 10$ and $p = 0.5$. We then perform the behavior graph construction with both methods, and we measure the memory consumption derived from the transitive reduction method (keeping in memory one behavior graph for each uncertain trace) versus the improved method (which generates a multiset of behavior graphs, one for each variant in the uncertain log).

\begin{figure}[H]
	\centering
	\includegraphics[width=.8\textwidth, keepaspectratio]{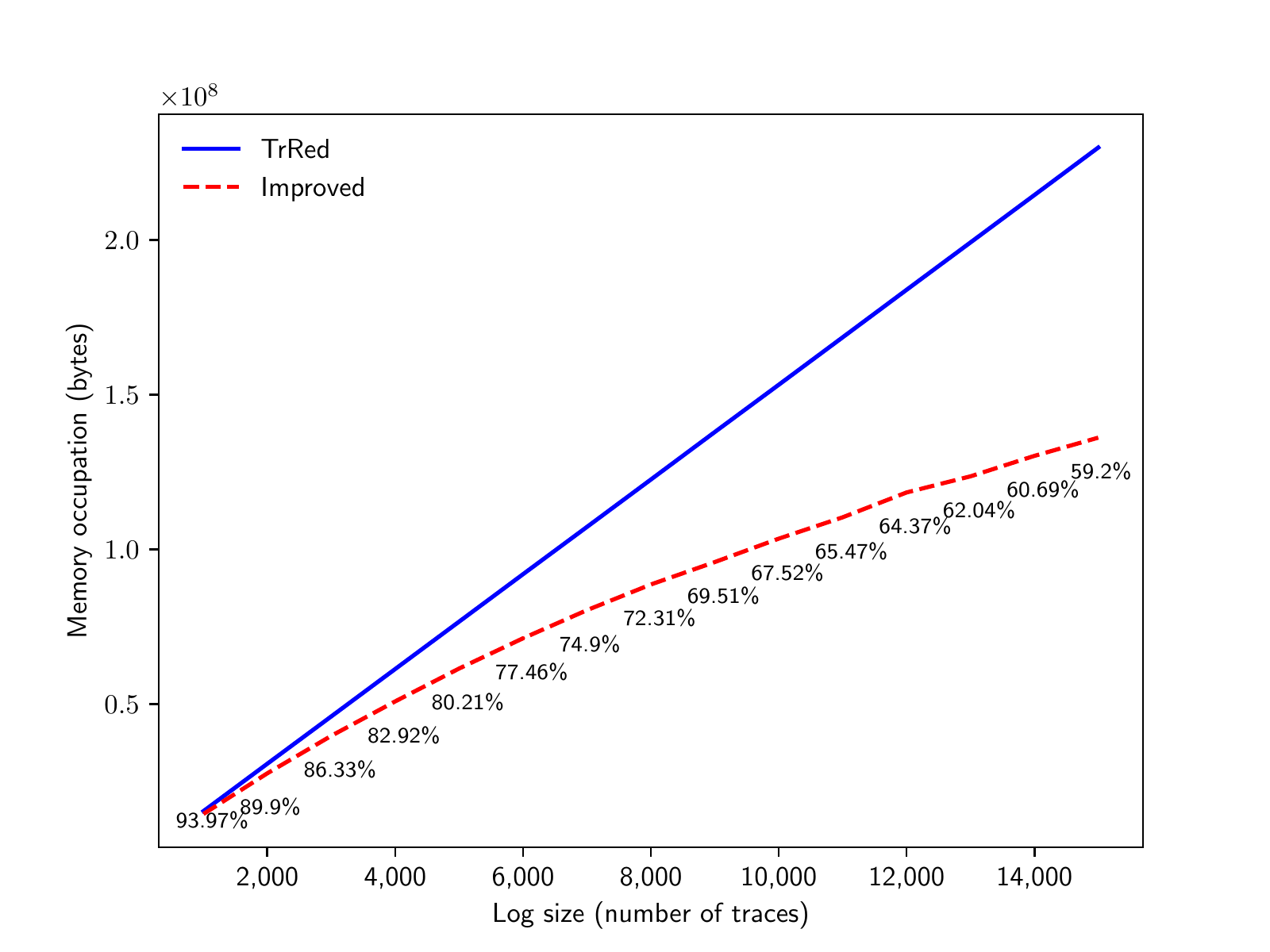}
	\caption{Memory consumption in bytes needed to store the behavior graphs for synthetic uncertain event logs with traces of length $l = 10$ events and timestamp uncertainty of $p = 0.5$, with an increasing number of traces $n$.}
	\label{fig:mem}
\end{figure}

The results are summarized in Figure~\ref{fig:mem}. Note that when $n$ increases, more and more uncertain traces are characterized by the same behavior graph, and can then be grouped in the same variant. This allows the improved algorithm to store the uncertain log more effectively. At $n = 15000$, the space needed by the multiset of behavior graphs is $59.2\%$, a sizable improvement in memory requirements when analyzing uncertain event logs of substantial dimensions. This improvement in memory consumption is a consequence of the new technique utilized in this paper to obtain the timestamp list, which enables such refinement with respect to the technique illustrated in~\cite{pegoraro2020efficient}.

Finally, to elucidate research question \emph{Q5} we compared the computation time for behavior graphs creation on real-life event logs, where we artificially inserted timestamp uncertainty in progressively higher percentage of uncertain events $p$ as described for the experiments above. We considered three event logs: an event log tracking the activities of the help desk process of an Italian software company, a log related to the management of road traffic fines in an Italian municipality, and a log from the BPI Challenge 2012 related to a loan application process. The results, presented in Figure~\ref{fig:reallifeprobs}, closely adhere to the findings of the experiments with synthetically generated uncertain event data: the novel method provides a substantial speedup, that remains rather stable with respect to the percentage $p$ of uncertain events added in the log.

\begin{figure}[H]
	\centering
	\includegraphics[width=.8\textwidth, keepaspectratio]{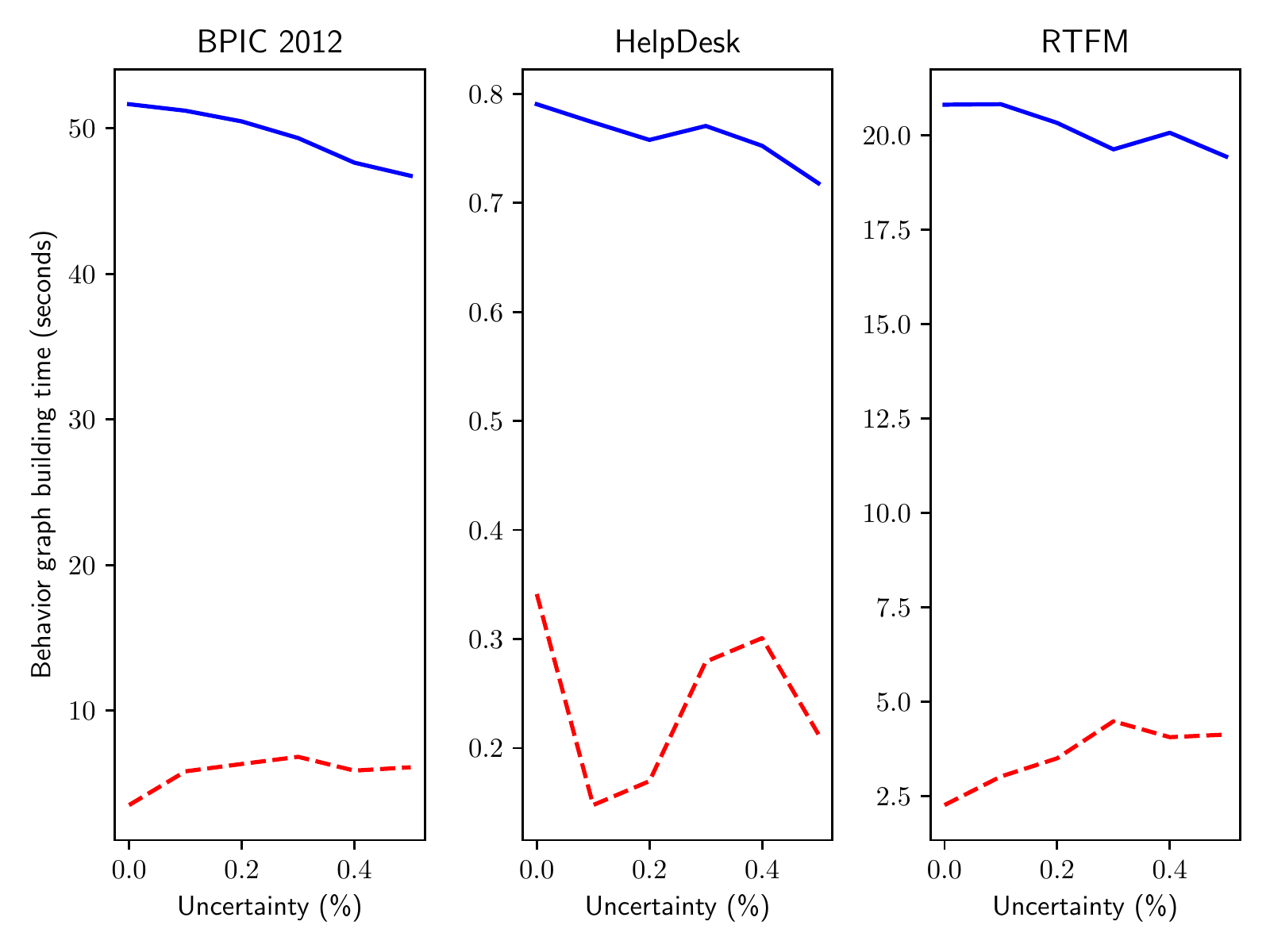}
	\caption{Execution times in seconds for real-life event logs with increasing percentages $p$ of timestamp uncertainty.}
	\label{fig:reallifeprobs}
\end{figure}

\subsection{Applications of the Behavior Graph Construction}

In Section~\ref{sec:introduction} we saw how building the behavior graph is a fundamental preprocessing step for both process discovery and conformance checking when dealing with uncertain event logs. In the previous section, we showed in practice how the novel algorithm presented in this paper impacts the computation time for the construction of behavior graphs. Now, let us have a glance into the effect of the speedup when applied to process mining techniques.

In this additional experiment we consider the conformance checking problem. In~\cite{pegoraro2019mining} we proposed an approach to compute upper and lower bounds for the conformance score of a trace against a reference Petri net through the alignment technique, which yields alignments for the best- and worst-case scenarios of an uncertain trace as illustrated in Section~\ref{sec:introduction}. The experiment is set up to assess the effect of the new behavior graph construction on the overall performance of conformance checking over uncertain data. We first generate a Petri net with $t$ transitions, simulate a log by playing out $n = 500$ traces, and add timestamp uncertainty with $p = 0.1$. We then compute the lower bound for conformance between the uncertain event log and the Petri net used as a source, and compare the overall execution time for conformance using the two different methods for the creation of the behavior graph. In this specific experiment, we also considered the other types of uncertainty in process mining illustrated in the taxonomy of~\cite{pegoraro2019mining}, as well as all types of uncertainty simulate on the same log.

\begin{figure}[H]
	\centering
	\includegraphics[width=\textwidth, keepaspectratio]{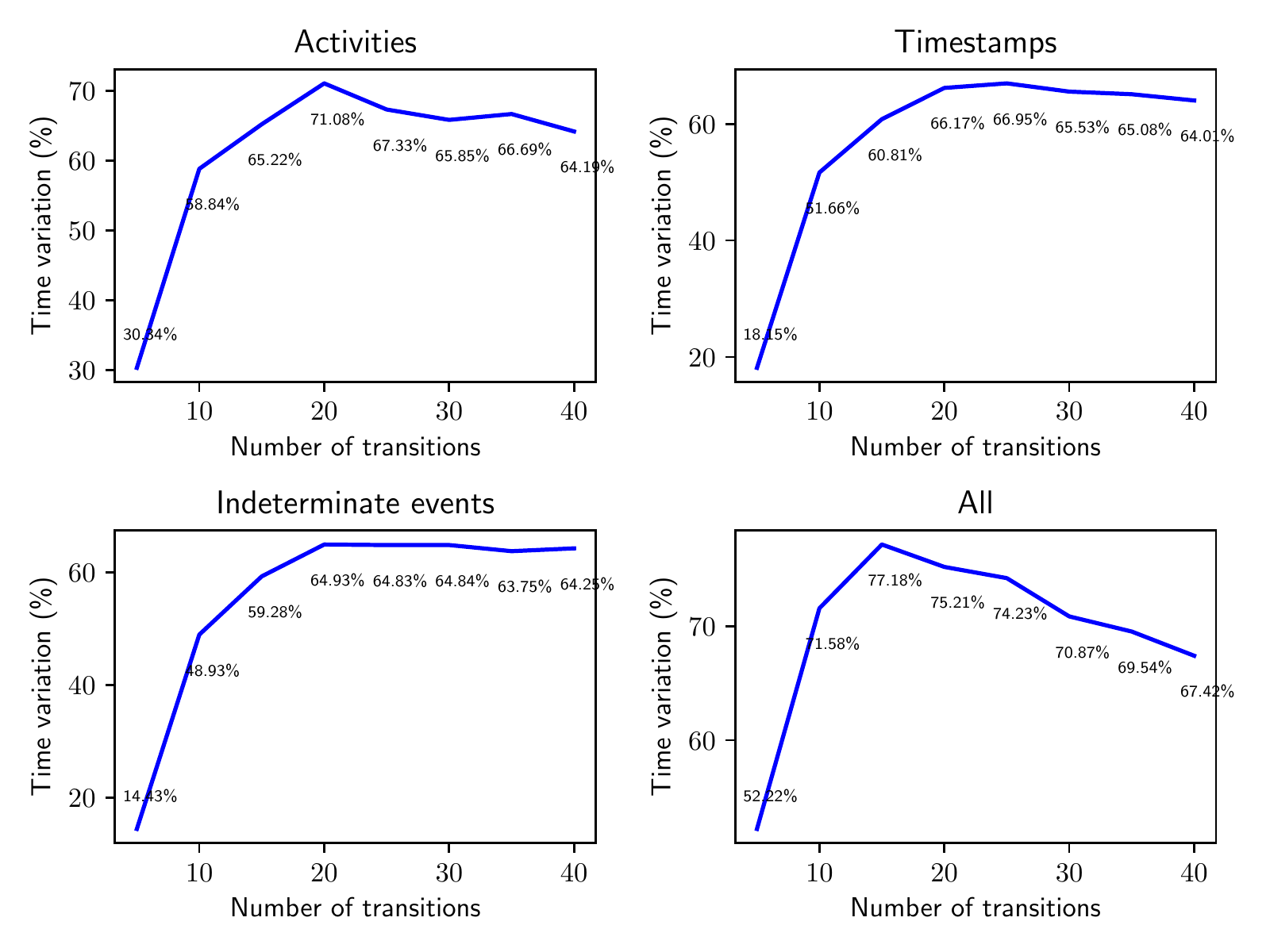}
	\caption{Relative variation in computation time obtained through the improved behavior graph construction when applied to the computation of conformance bounds between a synthetic uncertain log and a Petri net with an increasing number of transitions. The synthetic uncertain logs have $n = 500$ traces and timestamp uncertainty has been introduced with $p = 0.1$.}
	\label{fig:cc}
\end{figure}

The results are shown in Figure~\ref{fig:cc}. We can see that, on very small nets ($t = 5$), the alignment algorithm takes a short time to execute, so the speedup provided by the improved behavior graph construction has a larger impact on the total computation time (taking as little as $30.71\%$ of the time to calculate alignments). With the increase of $t$, the computation time for conformance checking using the fast construction of the behavior graph appears to stabilize around $65\%$ of the time needed if we employ the na{\"i}ve construction when considering only one type of uncertainty in isolation. This translates in a reduction of roughly $35\%$ of computation time for the very common problem of calculating the conformance score between event data and a reference model, a significant impact on performances of concrete applications of process mining over uncertain data. When compounding all types of uncertainty we see a similar effect, although for $t = 5$ the improved method takes $52.22\%$ of the time required by the baseline construction, a less dramatic effect than the other uncertainty settings. This is due to the fact that even at such small scales, the high number of realizations of traces slow down the alignment phase in the computation.

In evaluating this result, it is important to consider that alignments are a notoriously time-intensive technique~\cite{lee2017replay}, since the technique is based on an $\text{A}^*$ search on a state space that consists in pairs of the activities in the trace combined with the possible actions in the model. As a consequence, the impact of the algorithm presented in this paper is limited by the characteristics of the implementation of such alignment technique; combining it with more refined alignment algorithms would further improve the gain in speed.

In summary, the outcomes of the experiments show how our new algorithm hereby presented outperforms the previous method for creating the behavior graph on all the parameters in which the problem instance can scale in dimensions, in both the time and space dimensions. The experiment designed to answer \emph{Q3} shows that, like the na{\"i}ve algorithm, our novel method being is essentially insensitive to the percentage of events with uncertain timestamps contained in a trace. This fact is also verified by the experiment associated with \emph{Q5} on real-life data with added time uncertainty. While for every combination of parameters we benchmarked the novel algorithm runs in a fraction of time required by the baseline method, the experiments also confirm the improvements in asymptotic time complexity demonstrated through theoretical complexity analysis.

\section{Related Work}\label{sec:related}
The topic of process mining analysis over uncertain event data is relatively new, and little research has been carried out. The work introducing the concept of uncertainty in process mining, together with a taxonomy of the various types of uncertainty, specifically illustrated that if a trace displays uncertain attributes, it contains behavior, which can be effectively represented through graphical models -- namely, behavior graphs and behavior nets~\cite{pegoraro2019mining}. Differently to classic process mining, where we have a clearly defined separation between data and model and between the static behavior of data and the dynamic behavior of models, the distinction between data and models becomes more unclear in presence of uncertainty, because of the variety in behavior that affects the data. Representing traces through process models is utilized in~\cite{pegoraro2019mining} for the computation of upper and lower bounds for conformance scores of uncertain process traces against classic reference models. Another practical application of behavior graphs in the field of process mining over uncertain event data is presented in~\cite{pegoraro2019discovering}. Behavior graphs of uncertain traces are employed to determine the number of possible directly-follows relationships between uncertain events, with the end goal of automatically discovering process models from uncertain event data.

Albeit, as said, the application of the concept of uncertainty in data to process mining is recent, the same idea has precedents in the older field of data mining. Aggarwal and Philip~\cite{aggarwal2008survey} offer an overview of the topic of uncertain data and its analysis, with a strong focus on querying. Such data is modeled on the basis of probabilistic databases~\cite{suciu2011probabilistic}, a foundational concept in the setting of uncertain data mining. A branch of data mining particularly close to process mining is frequent itemsets mining: an efficient algorithm to search for frequent itemsets over uncertain data, the U-Apriori, have been described by Chui et al.~\cite{chui2007mining}.

Behavior graphs are Directed Acyclic Graphs (DAGs), which are widely used throughout many areas of science to represent with a graph-like model dependencies, precedence relationships, time information, or partial orders. They are effectively utilized in circular dependency analysis in software~\cite{al2014shape}, probabilistic graphical models~\cite{bayes1763lii}, dynamic graphs analytics~\cite{mariappan2019graphbolt}, and compiler design~\cite{aho2007compilers}. In process mining, \emph{Conditional Partial Order Graphs} (CPOGs) -- which consist of collections of DAGs -- have been exploited by Mokhov et al.~\cite{mokhov2016mining} to aid the task of process discovery.

We have seen throughout the paper that uncertainty on the timestamp dimension -- namely, representing at which time an event occurred with an interval of possible timestamps -- generates, on the precedence relationships of events, a partial order. Although uncertainty research in process mining provides a novel justification of partial ordering that spawns from specific attribute values, the idea of having a partial order instead of a total order among events in a trace has precedents in process mining research. Lu et al.~\cite{lu2014conformance}\cite{lu2014checking} examined the problem of conformance checking through alignments in the case of partially ordered traces, and developed a construct to represent conformance called a \emph{p-alignment}. Genga et al.~\cite{genga2018discovering} devised a method to identify highly frequent anomalous patterns in partially ordered process traces. More recently, van der Aa et al.~\cite{van2020partial} developed a probabilistic infrastructure that allows to infer the most likely linear extension of a partial order between events in a trace, with the goal of ``resolving'' the partial order.

An important aspect to notice is that conformance checking over uncertain event data is not to be confused with stochastic conformance checking, which concerns measuring conformance of certain event data against models enriched with probabilistic information. The probabilities decorating a stochastic model do not derive from uncertainties in event data, but rather from frequency of activities~\cite{leemans2020stochastic} or from performance indicators~\cite{rogge2013discovering}.

A review of related work on the topic of the asymptotic complexity of the transitive reduction and the equivalent problem of matrix multiplication is provided with the complexity analysis of the algorithms examined by this paper, in Section~\ref{sec:proofs}.

\section{Conclusions}\label{sec:conclusions}
The creation of the behavior graphs -- a graphical structure of paramount importance for the analysis of uncertain data in the domain of process mining -- plays a key role as initial processing step for both conformance checking and process discovery of process traces containing events with timestamp uncertainty, the most critical type of uncertain behavior. It allows, in fact, to represent the time relationship between uncertain events, which can be in a partial order. The behavior graph also carries the information regarding other types of uncertainty, like uncertain activity labels and indeterminate events. Such a representation is vital to establish which possible sequence of events in an uncertain trace most adhere to the behavior prescribed by a reference model, thereby enabling conformance checking; and to measure the number of possible occurrences of the directly-follows relationship between activities in an event log, making process discovery over uncertainty possible. Extracting behavior graphs from uncertain event data is thus concomitantly crucial and time consuming. In this paper, we show an improvement for the performance of uncertainty analysis by proposing a new algorithm that enables the creation of behavior graphs in quadratic time in the number of events in the trace. This novel method additionally allows for the representation of an uncertain log as a multiset of behavior graphs, which relevance is twofold: it allows to represent the control-flow information of an uncertain event log in a more compact manner by using less memory, and naturally extends the concept of variant -- central throughout the discipline of process mining -- to the uncertain domain. We proved the correctness of this novel algorithm, we showed asymptotic upper and lower bounds for its time complexity, and implemented performance experiments for this algorithm that effectively show the gain in computing speed it entails in real-world scenarios.

\bibliographystyle{splncs04}
\bibliography{bibliography}


\end{document}